\documentclass[11pt]{article}
\usepackage{amssymb,amsthm,fullpage,verbatim}
\newtheorem{lemma}{Lemma}
\newtheorem{fact}{Fact}
\newtheorem{theorem}{Theorem}

\newcommand{\chan}{\mathfrak{C}}

\newcommand{\cD}{\mathcal{D}}
\newcommand{\cH}{\mathcal{H}}
\newcommand{\cS}{\mathcal{S}}
\newcommand{\cB}{\mathcal{B}}
\newcommand{\cX}{\mathcal{X}}
\newcommand{\cY}{\mathcal{Y}}
\newcommand{\cZ}{\mathcal{Z}}
\newcommand{\cQ}{\mathcal{Q}}
\newcommand{\cU}{\mathcal{U}}
\newcommand{\cV}{\mathcal{V}}

\newcommand{\hcD}{\hat{\cD}}

\newcommand{\codebook}{\mathbf{C}}

\newcommand{\err}{\mathrm{err}}
\newcommand{\avgerr}{\overline{\mathrm{err}}}

\newcommand{\hPi}{\hat{\Pi}}
\newcommand{\tPi}{\tilde{\Pi}}
\newcommand{\htPi}{\hat{\tilde{\Pi}}}

\newcommand{\E}{{\rm {\bf E}}\,}
\newcommand{\Tr}{{\rm Tr}\;}
\newcommand{\rank}[1]{{\rm rank}(#1)}

\newcommand{\id}{\leavevmode\hbox{\small1\kern-3.8pt\normalsize1}}

\newcommand{\mm}{\rule[2pt]{1mm}{1pt}\,\rule[2pt]{1mm}{1pt}}

\newcommand{\elltwo}[1]{\left\|{ #1 }\right\|_2}
\newcommand{\ellone}[1]{\left\|{ #1 }\right\|_1}
\newcommand{\ellinfty}[1]{\left\|{ #1 }\right\|_\infty}

\newcommand{\ket}[1]{| #1 \rangle}
\newcommand{\bra}[1]{\langle #1 |}
\newcommand{\ketbra}[1]{\ket{#1}\bra{#1}}
\newcommand{\braket}[2]{\langle {#1} \ket{#2}}

\title{{\bf Achieving the Han-Kobayashi inner bound for the quantum
interference channel by sequential decoding}}

\author{
Pranab Sen\thanks{
School of Technology and Computer Science, Tata Institute of Fundamental
Research, Mumbai 400005, India.
}
}

\date{}

\begin{document}
\maketitle
\begin{abstract}
In this paper, we study the power of sequential decoding strategies
for several channels with classical input and quantum output.
In our sequential decoding strategies, the receiver loops 
through all
candidate messages trying to project the received state onto a
`typical' subspace for the candidate message under consideration, 
stopping if the projection succeeds for a message, which is then declared
as the guess of the receiver for the sent message. 
We show that even such a conceptually simple strategy can be used to 
achieve rates up to the
mutual information for a single sender single receiver channel called
cq-channel henceforth, as
well as the standard inner bound for a two sender single receiver
multiple access channel, called ccq-MAC in this paper. 
Our decoding scheme for the ccq-MAC uses a new kind of conditionally
typical projector which is constructed using a geometric result about
how two subspaces interact structurally.

As the main application of our methods, we construct an encoding and 
decoding scheme achieving the 
Chong-Motani-Garg inner bound for a two sender two receiver
interference channel with classical input and quantum output, called
ccqq-IC henceforth. This matches the best known inner bound for the 
interference
channel in the classical setting. Achieving the Chong-Motani-Garg inner
bound, which is known to be equivalent to the Han-Kobayashi inner 
bound,
answers an open question raised recently by Fawzi et al.~\cite{mcgill:qic}.
Our encoding scheme is the same as that of Chong-Motani-Garg, and our
decoding scheme is sequential. 
\end{abstract}

\section{Introduction}
For many classical channels, one can construct a simple sequential
decoding strategy as follows: the receiver loops through all
candidate messages trying to check if the current candidate message
is jointly typical with the received string. If yes, the receiver
stop and declares the current candidate message as its guess for the 
sent message. If no, the receiver tries the next candidate message and
so on. If all candidate messages are exhausted, the receiver declares
failure. Such a simple sequential decoding strategy works for example 
for the
single sender single receiver channel as well as the many sender single
receiver multiple access channel.

For a quantum channel, one can ask if the sequential decoding intuition
from the classical world leads to a valid decoding strategy. To
simplify the setting, we shall restrict ourselves to channels with
classical input and quantum output. Note that an encoding and decoding 
strategy for such a channel immediately gives us, via standard 
arguments, an encoding and
decoding strategy for sending classical information over channels with
quantum input and quantum output without any entanglement assistance.
For a single sender single receiver unassisted classical-quantum channel,
called cq-channel henceforth, Winter~\cite{winter:strongconverse}
gave a sequential decoding strategy that achieves rates up to
the mutual information. Winter's decoding strategy used a sequence
of POVMs. Recently, Giovannetti, Lloyd and Maccone~\cite{lloyd:seq}
improved Winter's result by constructing a decoding strategy for
the cq-channel that applies a sequence of orthogonal projections instead
of a sequence of POVMs. Giovannetti et al.'s decoder tries to project
the sent message alternately onto the typical subspace of the average
message and the typical subspace of the current candidate message.
This leaves us wondering if a simpler sequential decoding strategy
is possible, in particular, whether just projecting the received 
state onto the typical subspace of the current candidate message is good
enough for decoding. Besides the two sequential decoding results
mentioned above, other decoders for the 
cq-channel~\cite{holevo,schumacherwestmoreland,hayashinagaoka}
use conceptually
more complex strategies including in particular, the pretty good
measurement also known as the square-root measurement~\cite{pgm}.

In this paper, we first construct a sequential decoding scheme for
the cq-channel where the decoder just tries to project the received state
onto the typical subspace of the current candidate message. We show
that even this simple scheme can achieve rates up to the mutual 
information.
The scheme is simpler and possibly more efficient than the ones of 
Giovanetti et al., Winter, as well as other known schemes, and arguably
sheds more light into the interplay
between various typical projectors used in the decoding process.
Also, our proof of correctness is clearer and bears
a close resemblance to standard proofs of correctness of the decoder
in the classical setting. Interestingly, our correctness proof 
does not use the pretty good measurement standard
in many quantum Shannon theory papers, and may lead to more efficient
decoders.

We then construct a sequential decoder for the two sender single
receiver multiple access channel with classical input and quantum output,
called ccq-MAC henceforth, achieving the standard inner bound for
this channel. For this, we have to replace the
standard conditionally typical projectors by a new kind of conditionally
typical projector approximating them. The existence of
such suitable projectors follows from a powerful geometric fact about how
two subspaces interact. Our sequential decoding strategy for the
ccq-MAC is `jointly typical' or `simultaneous'. Until very recently,
the only decoders known for the ccq-MAC were of the `successive
cancellation' variety (e.g.~\cite{winter:mac}. Simultaneous decoders are 
crucial in proving
many results in classical multiuser information theory. The lack of
simultaneous decoders in the quantum setting is a major stumbling
block in the development of multiuser quantum information theory.
Our simultaneous decoder for the ccq-MAC should be an important first step
towards the development of a powerful multiuser quantum information theory.
Another simultaneous decoder for the ccq-MAC was discovered very recently
by Fawzi et al.~\cite{mcgill:qic}. 

Finally, as the main application of our methods, we construct an 
encoding and decoding scheme achieving the 
Chong-Motani-Garg~\cite{cmg} inner bound for a two sender two receiver
interference channel with classical input and quantum output, called
ccqq-IC henceforth. The Chong-Motani-Garg inner bound is known to
be equivalent to the Han-Kobayashi~\cite{hk} inner bound, which is
the best known inner bound for the interference channel in the classical
setting.
This answers an open question raised recently by 
Fawzi et al.~\cite{mcgill:qic}. The Chong-Motani-Garg and
Han-Kobayashi inner bounds are not known to be achievable by
successive cancellation. We prove the Chong-Motani-Garg inner bound
for the ccqq-IC by analysing, as in the classical case, a special three 
sender single 
receiver classical-quantum multiple access channel, called CMG-MAC
henceforth. We prove an inner bound for the CMG-MAC using our techniques
which turns out
to be larger than Chong, Motani and Garg's classical
inner bound for the same encoding strategy! However, this improved 
inner bound does not
enlarge the Chong-Motani-Garg region for the interference channel.
Our encoding scheme is the same as that of Chong-Motani-Garg, and our
decoding scheme is sequential. 

\subsection{Organisation of the paper}
We first discuss some preliminary facts in Section~\ref{sec:prelim}.
In Section~\ref{sec:cq}, as a demonstration of our basic sequential 
decoding approach, we demonstrate the achievability of 
rates up to the 
mutual information for a single sender single receiver channel.
In Section~\ref{sec:MAC}, we construct our sequential simultaneous
decoder for the ccq-MAC and the CMG-MAC.
In Section~\ref{sec:ccqq} we use our inner bound for the 
CMG-MAC to prove the Chong-Motani-Garg
inner bound for the ccqq-IC. We conclude with a brief discussion in
Section~\ref{sec:discussion}.

\section{Preliminaries}
\label{sec:prelim}
\subsection{Typical projectors}
\label{subsec:typical}
Let $\cX$ be a finite set.
We shall overload the symbol $\cX$ to also denote the Hilbert 
space with computational bases 
being the letters from the alphabet $\cX$.
Let $B$ be a quantum system
with its accompanying Hilbert space denoted by $\cB$. 
The symbol $|\cX|$ will denote the size
of the set $\cX$ which is also the dimension of the Hilbert 
space $\cX$.
Similarly, the symbol $|\cB|$ will denote the dimension of the Hilbert
space $\cB$.

Fix a probablity distribution $p_X$ on $\cX$.
Let $X$ denote the corresponding random variable.
The entropy of $X$ is defined as
$H(X) := -\sum_{x \in \cX} p_X(x) \log p_X(x)$.
Let $X^n$ denote the random variable corresponding to $n$ 
independent copies of $X$. For a positive integer $n$, 
the notation $x^n$ will stand
for a sequence of length $n$ over the alphabet $\cX$.
Let $N(x|x^n)$ denote the number of positions in 
the sequence $x^n$ containing symbol $x$.
Let $0 < \delta < 1$.
The set $T^{X^n}_\delta$ of {\em frequency
typical sequences} of length $n$ over the alphabet 
$\cX$ for the distribution $p_X$ is defined as follows:
\[
T^{X^n}_\delta := 
\left\{x^n: \forall x \in \cX 
       \left|\frac{N(x | x^n)}{n} - p_X(x)\right| \leq p_X(x) \delta 
\right\}.
\]
The above definition is the one used in e.g.~El Gamal and Kim's lecture
notes~\cite{elgamalkim}.
Note that other works often give a slightly different definition of
 $\delta$-typical viz. the {\em absolute error} per symbol is at most
$\delta$. We shall however use the above definition for relative error
$\delta$ as it is more suitable for our purposes.
The frequency typical set satisfies the following properties.
\begin{fact}
\label{fact:typicalset}
Let $0 < \epsilon, \delta < 1/2$. 
Define $p_{\mathrm{min}} := \min_{x \in \cX, p_X(x)>0} p_X(x)$.
Let $n \geq 2 p_{\mathrm{min}}^{-1} \delta^{-2} \log(|\cX| / \epsilon)$.
Define $c(\delta) := \delta \log |\cX| - \delta \log \delta$.
Then,
\begin{eqnarray*}
\sum_{x^n \in T^{X^n}_{\delta}} p_{X^n}(x^n) & \geq & 1 - \epsilon, \\
\forall x^n \in T^{n,\delta}_p: 2^{-n(H(X) + c(\delta))} & \leq &
p_{X^n}(x^n) \leq 2^{-n(H(X) - c(\delta))}, \\
|T^{X^n}_\delta| & \leq & 2^{n(H(X) + c(\delta))}.
\end{eqnarray*}
\end{fact}

Given a density matrix $\rho$ over Hilbert space $\cB$, consider a 
canonical eigenbasis $\cS = \{s_1, \ldots, s_{|\cB|}\}$ of $\rho$. 
Consider the diagonalisation $\rho = \sum_{s \in \cS} q_S(s) \ketbra{s}$. 
Thinking of $\{q_S(s)\}_s$ as a probability
distribution over $\cS$, we can analogously define the entropy of
$\rho$ as $H(\rho) := -\Tr [\rho \log \rho] = H(q_S)$, and
the frequency typical subspace
of $\cB^{\otimes n}$ corresponding to the $n$-fold tensor power operator
$\rho^{\otimes n}$, as follows:
\[
T^{\rho^n}_\delta := \mathrm{span} 
\left\{s^n: \forall s \in \cS 
            \left|\frac{N(s | s^n)}{n} - q_S(s)\right| \leq \delta q_S(s)
\right\}.
\]
Let $\Pi^{\rho^n}_{\delta}$ denote the orthogonal projection onto
$T^{\rho^n}_\delta$. Observe that $\Pi^{\rho^n}_\delta$ commutes
with $\rho^{\otimes n}$. The typical projector satisfies 
properties similar to those of 
Fact~\ref{fact:typicalset}~\cite{schumacher,bennettentanglement}.
\begin{fact}
\label{fact:typicalspace}
Let $0 < \epsilon, \delta < 1/2$. Let $\rho$ be a quantum state.
Define 
\[
q_{\mathrm{min}} := q_{\mathrm{min}}(\rho) 
        := \min_{s \in \cS, q_S(s)>0} q_S(s).
\]
Suppose that 
$n \geq 2 q_{\mathrm{min}}^{-1} \delta^{-2} \log(|\cB| / \epsilon)$.
Define $c(\delta) := \delta \log |\cB| - \delta \log \delta$. Then,
\begin{eqnarray*}
\Tr [\rho^{\otimes n} \Pi^{\rho^n}_\delta] & \geq & 1 - \epsilon, \\
2^{-n(H(\rho) + c(\delta))} \Pi^{\rho^n}_\delta & \leq &
 \Pi^{\rho^n}_\delta \rho^{\otimes n} \leq 
2^{-n(H(\rho) - c(\delta))}  \Pi^{\rho^n}_\delta, \\
\Tr \Pi^{\rho^n}_\delta & \leq & 2^{n(H(\rho) + c(\delta))}.
\end{eqnarray*}
\end{fact}

Suppose $\{\rho_x\}_{x \in \cX}$ is an ensemble of quantum states
in $\cB$.
For an input sequence $x^n$, define 
$\rho_{x^n} := \rho_{x_1} \otimes \cdots \otimes \rho_{x_n}$. 
After reordering multiplicands, $\rho_{x^n}$ can be written as
$\rho_{x^n} = \bigotimes_{x \in \cX} \rho_x^{\otimes N(x | x^n)}$. 
We can then define the frequency typical conditional 
projector~\cite{holevo,schumacherwestmoreland,winterthesis} 
$\Pi^{\rho_{x^n}}_\delta$ as
\[
\Pi^{\rho_{x^n}}_\delta :=
\bigotimes_{x \in \cX} \Pi^{\rho_x^{\otimes N(x | x^n)}}_\delta.
\]
Let $H(B|X)$ denote the conditional Shannon entropy of the quantum
system $B$ given random variable $X$ i.e.
$H(B|X) := \sum_{x \in \cX} p_X(x) H(\rho_x)$.
Note that $\Pi^{\rho_{x^n}}_\delta$ commutes with $\rho_{x^n}$.
The conditionally typical projector satisfies properties similar to
Fact~\ref{fact:typicalspace} if the input $x^n$ is typical.
\begin{fact}
\label{fact:condtypicalspace}
Let $0 < \epsilon, \delta < 1/2$. 
Suppose $X$ is a classical random variable and $\rho^{XB}$ is a 
classical-quantum state. Define 
\[
p_{\mathrm{min}} := \min_{x \in \cX, p_X(x) > 0} p_X(x), ~~~
q_{\mathrm{min}} := \min_{x \in \cX, p_X(x) > 0} 
 q_{\mathrm{min}}(\rho_x).
\]
Let $n \geq 4 \delta^{-2} p_{\mathrm{min}}^{-1} q_{\mathrm{min}}^{-1}
     \log(|\cB| |\cX| / \epsilon)$.
Let $x^n \in T^{X^n}_\delta$.
Define $c(\delta) := \delta \log |\cB||\cX| - \delta \log \delta$. Then,
\begin{eqnarray*}
\Tr [\rho_{x^n} \Pi^{\rho_{x^n}}_\delta] & \geq & 1 - \epsilon, \\
2^{-n(H(B|X) + c(\delta))} \Pi^{\rho_{x^n}}_\delta & \leq &
 \Pi^{\rho_{x^n}}_\delta \rho_{x^n} \leq 
2^{-n(H(B|X) - c(\delta))}  \Pi{\rho_{x^n}}_\delta, \\
\Tr \Pi^{\rho_{x^n}}_\delta & \leq & 2^{n(H(B|X) + c(\delta))}.
\end{eqnarray*}
\end{fact}

Define the average state of the ensemble 
$\rho := \sum_{x \in \cX} p_X(x) \rho_x$.
The following nontrivial fact that the typical projector corresponding to 
the average
state $\rho^{\otimes n}$ actually captures most of the mass of 
$\rho_{x^n}$ if $x^n$ is typical was proved by 
Winter~\cite{winter:strongconverse}. We will require
a slightly stronger statement for our purposes which we now state.
\begin{fact}
\label{fact:winter}
Let $0 < \epsilon, \delta < 1/2$. Suppose $X$, $Y$ are classical
random variables 
and $\rho^{XYB}$ is a classical-classical-quantum state.
Define 
\[
p_{\mathrm{min}} := \min_{(x,y) \in (\cX \times \cY), p_{XY}(x,y) > 0} 
                        p_{XY}(x,y), ~~~
q_{\mathrm{min}} := \min_{(x,y) \in (\cX \times \cY), p_{XY}(x,y) > 0} 
                        q_{\mathrm{min}}(\rho_{xy}).
\]
Let $n \geq 4 \delta^{-2} p_{\mathrm{min}}^{-1} q_{\mathrm{min}}^{-1}
     \log(|\cB| |\cX| / \epsilon)$.
Let $x^n y^n \in T^{(XY)^n}_\delta$.
Then,
\[
\Tr \left[\rho_{x^n y^n} \Pi^{\rho^{\otimes n}}_{2\delta}\right] 
\geq 1 - \epsilon, ~~~~ 
\Tr \left[\rho_{x^n y^n} \Pi^{\rho_{x^n}}_{6\delta}\right] 
\geq 1 - \epsilon.
\]
\end{fact}

A point to note about the above facts on typical subspaces is that they
are all proved using Chernoff bounds instead of Chebyshev bounds or
law of large numbers as is sometimes done in the literature. This is
done for reasons of efficiency viz. to ensure that the concentration of 
measure $\epsilon$ is
exponentially small in the number of copies $n$. This exponential 
dependence is not very important however, and we could have used
weaker concentration results for our purposes. 
Note also that
in all cases above, $c(\delta) \rightarrow 0$ and 
$\delta^{-2} c(\delta) \rightarrow \infty$ as $\delta \rightarrow 0$.

\paragraph{Notation:}
Finally, we fix some notation for entropic quantities that will be
used throughout the paper. Suppose $A B C$ is a quantum (which also
subsumes classical) system with joint density matrix $\rho^{ABC}$. 
The entropy of, for example, subsystem $A B$ under joint density
matrix $\rho^{ABC}$ is denoted by
$H(AB)_\rho := H(\rho^{AB})$, where $\rho^{AB}$ is the marginal density
matrix of subsystem $AB$. Sometimes when the underlying state $\rho$
is clear from the context, we drop it in the subscript.
The conditional entropy, for example, of
system $C$ given system $AB$ is denoted by
$H(C|AB)_\rho := H(ABC) - H(AB)$. Now suppose for example that
$B$ is classical, that is, the joint density matrix looks like
$\rho^{ABC} = \sum_b \Pr[b] \ketbra{b} \otimes \rho^{AC}_b$, where 
$\ket{b}$
ranges over the computational basis of Hilbert space $\cB$ of system
$B$, $\{\Pr[b]\}_b$ form a probability distribution over the computational
basis of $\cB$, and $\{\rho^{AC}_b\}_b$ are joint density matrices of 
system $AC$.
We will interpret $\rho^{AC}_b$ to be the joint density matrix of
system $AC$ when system $B$ is in the computational basis state $\ket{b}$,
which occurs with probability $\Pr[b]$.
If $B$ is classical, $H(C|AB)_\rho = \sum_b \Pr[b] H(C|A)_{\rho^{AC}_b}$.
We shall define $H(C|A, B=b) := H(C|A)_{\rho^{AC}_b}$.
The mutual information between, for example, $A$ and $BC$ is defined
as $I(A:BC)_\rho := H(A) + H(BC) - H(ABC)$. The mutual information of
$A$ and $C$ conditioned on $B$ is defined as
$
I(A:C|B)_\rho := H(A|B) + H(C|B) - H(AC|B) = 
H(AB) + H(CB) - H(ABC) - H(B).
$
If $B$ is classical, it is equal to
$
I(A:C|B)_\rho = \sum_b \Pr[b] I(A:C)_{\rho^{AC}_b}.
$
We shall define $I(A:C|B=b) := I(A:C)_{\rho^{AC}_b}$.

\subsection{Asymptotic smoothing}
\label{subsec:asymsmooth}

We can now use the facts about typical subspaces from 
Section~\ref{subsec:typical} to prove the following asymptotic
smoothing lemma. It will be applied in the analysis of the decoding
error probabilities of all our decoders henceforth. Some of these
applications can do with less general versions of the lemma, and it is
indeed possible to prove such versions with slightly better parameters.
But since they do not lead to significantly lower decoding error
probabilities, we stick to the following general version.
\begin{lemma}
\label{lem:asymsmooth}
Let $XZYB$ be a classical-classical-classical-quantum system with joint 
density matrix
\[
\rho^{XZYB} = 
\sum_{(x,z,y) \in \cX \times \cZ \times \cY} p_X(x) p_{Z|X}(z|x) p_Y(y)
\ketbra{x,z,y} \otimes \rho_{xzy}. 
\]
Define $\rho^{X^n Z^n Y^n B^n} := (\rho^{XZYB})^{\otimes n}$.
Let $0 < \epsilon, \delta < 1/64$. Let 
\begin{eqnarray*}
p_{\mathrm{min}} 
& := & \min_{(x,z,y) \in \cX \times \cZ \times \cY, p_{XZ}(x,z) p_Y(y) > 0}
       p_{XZ}(x,z) p_Y(y), \\
q_{\mathrm{min}} 
& := & \min_{(x,z,y) \in \cX \times \cZ \times \cY, p_{XZ}(x,z) p_Y(y) > 0}
       q_{\mathrm{min}}(\rho_{xzy}).
\end{eqnarray*}
Let $n \geq 4 \delta^{-2} p_{\mathrm{min}}^{-1} q_{\mathrm{min}}^{-1}
            \log(|\cB| |\cX| |\cZ| |\cY| / \epsilon)$.
Define $c(\delta) := \delta \log |\cB||\cX||\cZ||\cY|- \delta \log \delta$.
Then there is an ensemble of density matrices 
$\{\rho'^{B^n}_{x^n z^n y^n}\}_{(x^n, z^n, y^n) 
                \in \cX^n \times \cZ^n \times \cY^n}$ such that 
for 
\[
\rho'^{X^n Z^n Y^n B^n} := 
\sum_{(x^n, z^n, y^n) \in \cX^n \times \cZ^n \times \cY^n}
p_{X^n}(x^n) p_{Z^n|X^n}(z^n|x^n) p_{Y^n}(y^n) \ketbra{x^n, z^n, y^n} 
\otimes \rho'^{B^n}_{x^n z^n y^n},
\]
and 
\[
\rho'^{B^n}_{x^n z^n} := \sum_{y^n \in \cY^n} p_{Y^n}(y^n)
	                 \rho'^{B^n}_{x^n z^n y^n}, ~~
\rho'^{B^n}_{x^n} := \sum_{z^n \in \cZ^n} p_{Z^n|X^n}(z^n|x^n) 
                     \rho'^{B^n}_{x^n z^n}, ~~
\rho'^{B^n} := \sum_{x^n \in \cX^n} p_{X^n}(x^n) \rho'^{B^n}_{x^n},
\]
we have
\[
\begin{array}{l r c l}
\forall x^n z^n \in T^{(XZ)^n}_\delta: & 
\ellinfty{\rho'^{B^n}_{x^n z^n}} 
& \leq & 4 \cdot 2^{-n(H(B|XZ)_\rho - c(6\delta)}, \\ \\
\forall x^n \in T^{X^n}_\delta: & 
\ellinfty{\rho'^{B^n}_{x^n}} 
& \leq & 4 \cdot 2^{-n(H(B|X)_\rho - c(6\delta)}, \\ \\
& 
\ellinfty{\rho'^{B^n}} 
& \leq & 4 \cdot 2^{-n(H(B)_\rho - c(2\delta)}, \\ \\
\forall x^n z^n y^n \in T^{(XZY)^n}_\delta: & 
\ellone{\rho'^{B^n}_{x^n z^n y^n} - \rho^{B^n}_{x^n z^n y^n}}
& \leq & 11\sqrt{\epsilon}, \\ \\
& 
\ellone{\rho'^{X^n Z^n Y^n B^n} - \rho^{X^n Z^n Y^n B^n}} 
& \leq & 13\sqrt{\epsilon}.
\end{array}
\]
\end{lemma}
\begin{proof}
For $x^n \in \cX^n$ and $z^n \in \cZ^n$,
define 
\[
\rho_{x^n z^n} := \sum_{y^n \in \cY^n} p_{Y^n}(y^n) \rho_{x^n z^n y^n},
~~~~~
\rho_{x^n} := \sum_{z^n \in \cZ^n} p_{Z^n|X^n}(z^n|x^n) 
              \rho_{x^n z^n}.
\]
Recall that
$
\sum_{x^n \in \cX^n} p_{X^n}(x^n) \rho_{x^n} = \rho^{\otimes n},
$
where $\rho := \rho^B$.
Let $\Pi_{x^n z^n}$ denote the conditionally typical projector 
$\Pi^{\rho_{x^n z^n}}_{6\delta}$, $\Pi_{x^n}$ denote the conditionally
typical projector $\Pi^{\rho_{x^n}}_{6\delta}$ and $\Pi$ denote
the typical projector $\Pi^{\rho^n}_{2 \delta}$.
For $x^n z^n y^n \in T^{(XZY)^n}_\delta$, define
\[
\rho'_{x^n z^n y^n} := 
\frac{\Pi \Pi_{x^n} \Pi_{x^n z^n} \rho_{x^n z^n y^n} \Pi_{x^n z^n} 
          \Pi_{x^n} \Pi}
     {\Tr[\Pi \Pi_{x^n} \Pi_{x^n z^n} \rho_{x^n z^n y^n} \Pi_{x^n z^n}
              \Pi_{x^n} \Pi]}.
\]
From Facts~\ref{fact:winter} and \ref{fact:gentle}, we get
\begin{eqnarray*}
\lefteqn{
\Tr [\Pi \Pi_{x^n} \Pi_{x^n z^n} \rho_{x^n z^n y^n} \Pi_{x^n z^n} 
         \Pi_{x^n} \Pi] 
} \\
& \geq &
\Tr [\Pi \rho_{x^n z^n y^n} \Pi] -
\ellone{\rho_{x^n z^n y^n} - 
        \Pi_{x^n} \rho_{x^n z^n y^n} \Pi_{x^n}} \\
&      &
{} - 
\ellone{\Pi_{x^n} \rho_{x^n z^n y^n} \Pi_{x^n} - 
        \Pi_{x^n} \Pi_{x^n z^n} \rho_{x^n z^n y^n} \Pi_{x^n z^n}
        \Pi_{x^n}}\\
& \geq &
\Tr [\Pi \rho_{x^n z^n y^n} \Pi] -
\ellone{\rho_{x^n z^n y^n} - 
        \Pi_{x^n} \rho_{x^n z^n y^n} \Pi_{x^n}} -
\ellone{\rho_{x^n z^n y^n} - 
        \Pi_{x^n z^n} \rho_{x^n z^n y^n} \Pi_{x^n z^n}}\\
& \geq &
1 - \epsilon - 2\sqrt{\epsilon} - 2\sqrt{\epsilon}
\;\geq\;
1 - 5\sqrt{\epsilon},
\end{eqnarray*}
and
\begin{eqnarray*}
\lefteqn{\ellone{\rho_{x^n z^n y^n} - \rho'_{x^n z^n y^n}}} \\
& \leq &
\ellone{\rho_{x^n z^n y^n} - 
        \Pi \Pi_{x^n} \Pi_{x^n z^n} \rho_{x^n z^n y^n} \Pi_{x^n z^n} 
        \Pi_{x^n} \Pi} + 
\ellone{\Pi \Pi_{x^n} \Pi_{x^n z^n} \rho_{x^n z^n y^n} \Pi_{x^n z^n} 
            \Pi_{x^n} \Pi
        - \rho'_{x^n z^n y^n}} \\
& \leq &
\ellone{\rho_{x^n z^n y^n} - \Pi \rho_{x^n z^n y^n} \Pi} + 
\ellone{\Pi \rho_{x^n z^n y^n} \Pi - 
        \Pi \Pi_{x^n} \Pi_{x^n z^n} \rho_{x^n z^n y^n} \Pi_{x^n z^n} 
            \Pi_{x^n} \Pi} \\ 
&      & 
{} + 1 - \Tr [\Pi \Pi_{x^n} \Pi_{x^n z^n} \rho_{x^n z^n y^n} \Pi_{x^n z^n}
                  \Pi_{x^n} \Pi] \\
& \leq &
\ellone{\rho_{x^n z^n y^n} - \Pi \rho_{x^n z^n y^n} \Pi} + 
\ellone{\rho_{x^n z^n y^n} - 
        \Pi_{x^n} \Pi_{x^n z^n} \rho_{x^n z^n y^n} \Pi_{x^n z^n} 
        \Pi_{x^n}} \\
&      &
{} + 1 - \Tr [\Pi \Pi_{x^n} \Pi_{x^n z^n} \rho_{x^n z^n y^n} \Pi_{x^n z^n}
                  \Pi_{x^n} \Pi] \\
& \leq &
\ellone{\rho_{x^n z^n y^n} - \Pi \rho_{x^n z^n y^n} \Pi} + 
\ellone{\rho_{x^n z^n y^n} - \Pi_{x^n} \rho_{x^n z^n y^n} \Pi_{x^n}} + {}\\
&      &
\ellone{\Pi_{x^n} \rho_{x^n z^n y^n} \Pi_{x^n} - 
        \Pi_{x^n} \Pi_{x^n z^n} \rho_{x^n z^n y^n} \Pi_{x^n z^n} \Pi_{x^n}}
+ 1 - \Tr [\Pi \Pi_{x^n} \Pi_{x^n z^n} \rho_{x^n z^n y^n} \Pi_{x^n z^n}
                  \Pi_{x^n} \Pi] \\
& \leq &
\ellone{\rho_{x^n z^n y^n} - \Pi \rho_{x^n z^n y^n} \Pi} + 
\ellone{\rho_{x^n z^n y^n} - \Pi_{x^n} \rho_{x^n z^n y^n} \Pi_{x^n}} + {}\\
&      &
\ellone{\rho_{x^n z^n y^n} - 
        \Pi_{x^n z^n} \rho_{x^n z^n y^n} \Pi_{x^n z^n}}
+ 1 - \Tr [\Pi \Pi_{x^n} \Pi_{x^n z^n} \rho_{x^n z^n y^n} \Pi_{x^n z^n}
                  \Pi_{x^n} \Pi] \\
& \leq &
2\sqrt{\epsilon} + 2\sqrt{\epsilon} + 2\sqrt{\epsilon} + 5\sqrt{\epsilon} 
\;  = \;
11 \sqrt{\epsilon}.
\end{eqnarray*}
For $x^n z^n y^n \not \in T^{(XZY)^n}_\delta$, define 
$\rho'_{x^n z^n y^n} := \frac{\id}{|\cB|^n}$.
We now define 
\[
\rho'^{X^n Z^n Y^n B^n} := 
\sum_{(x^n,z^n,y^n) \in \cX^n \times \cZ^n \times \cY^n} 
p_{X^n}(x^n) p_{Z^n|X^n}(z^n|x^n) p_{Y^n}(y^n)
\ketbra{x^n, z^n, y^n} \otimes \rho'_{x^n z^n y^n}. 
\]
Then,
\begin{eqnarray*}
\lefteqn{\ellone{\rho'^{X^n Z^n Y^n B^n} - \rho^{X^n Z^n Y^n B^n}}}\\
&   =  & \sum_{(x^n,z^n,y^n) \in \cX^n \times \cZ^n \times \cY^n} 
         p_{(XZ)^n}(x^n z^n) p_{Y^n}(y^n) 
         \ellone{\rho'_{x^n z^n y^n} - \rho_{x^n z^n y^n}} \\
& \leq & \sum_{(x^n,z^n,y^n) \in T^{(XZY)^n}_\delta} 
         p_{(XZ)^n}(x^n z^n) p_{Y^n}(y^n) 
         \ellone{\rho'_{x^n z^n y^n} - \rho_{x^n z^n y^n}} + {} \\
&      & 2 \sum_{(x^n,z^n,y^n) \not \in T^{(XZY)^n}_\delta} 
         p_{(XZ)^n}(x^n z^n) p_{Y^n}(y^n) \\
&   <  & 11\sqrt{\epsilon} + 2\epsilon 
\;\leq\; 13\sqrt{\epsilon},
\end{eqnarray*}
where we used Fact~\ref{fact:typicalset}.

For $x^n z^n \in T^{(XZ)^n}_\delta$, we have
\begin{eqnarray*}
\rho'^{B^n}_{x^n z^n}
&   =  & \sum_{y^n: x^n z^n y^n \in T^{(XZY)^n}_\delta} p_{Y^n}(y^n)
         \frac{\Pi \Pi_{x^n} \Pi_{x^n z^n} \rho_{x^n z^n y^n} \Pi_{x^n z^n}
                   \Pi_{x^n} \Pi}
              {\Tr [\Pi \Pi_{x^n} \Pi_{x^n z^n} \rho_{x^n z^n y^n} 
                                  \Pi_{x^n z^n} \Pi_{x^n} \Pi]} + {} \\
&      & \sum_{y^n: x^n z^n y^n \not \in T^{(XZY)^n}_\delta} p_{Y^n}(y^n)
         \frac{\id}{|\cB|^n} \\
& \leq & (1- 5\sqrt{\epsilon})^{-1} 
         \sum_{y^n: x^n z^n y^n \in T^{(XZY)^n}_\delta} p_{Y^n}(y^n)
         \Pi \Pi_{x^n} \Pi_{x^n z^n} \rho_{x^n z^n y^n} \Pi_{x^n z^n} 
             \Pi_{x^n} \Pi +
         \frac{\id}{|\cB|^n} \\
& \leq & (1- 5\sqrt{\epsilon})^{-1} 
         \Pi \Pi_{x^n} \Pi_{x^n z^n} 
         \left(
         \sum_{y^n \in \cY^n} p_{Y^n}(y^n) \rho_{x^n z^n y^n} 
         \right) 
                       \Pi_{x^n z^n} \Pi_{x^n} \Pi +
         \frac{\id}{|\cB|^n} \\
&   =  & (1- 5\sqrt{\epsilon})^{-1} 
         \Pi \Pi_{x^n} \Pi_{x^n z^n} \rho_{x^n z^n} \Pi_{x^n z^n} \Pi_{x^n}
         \Pi + \frac{\id}{|\cB|^n} \\
& \leq & (1- 5\sqrt{\epsilon})^{-1} 2^{-n(H(B|XZ)_\rho - c(6\delta))}  
         \Pi \Pi_{x^n} \Pi_{x^n z^n} \Pi_{x^n} \Pi + \frac{\id}{|\cB|^n},
\end{eqnarray*}
where we used Facts~\ref{fact:winter} and
\ref{fact:condtypicalspace}. This shows that 
\[
\ellinfty{\rho'^{B^n}_{x^n z^n}} \leq 
(1-5\sqrt{\epsilon})^{-1} 2^{-n(H(B|XZ)_\rho - c(6\delta))} + 
|\cB|^{-n} \leq
4 \cdot 2^{-n(H(B|XZ)_\rho - c(6\delta))},
\]
since $\epsilon < 1/64$.

Similarly for $x^n \in T^{X^n}_\delta$, we have
\begin{eqnarray*}
\rho'^{B^n}_{x^n}
&   =  & \sum_{z^n \in \cZ^n} p_{Z^n|X^n}(z^n|x^n) \rho'^{B^n}_{x^n z^n} \\
& \leq & \sum_{z^n \in \cZ^n} p_{Z^n|X^n}(z^n|x^n) \left(
         (1- 5\sqrt{\epsilon})^{-1} 
         \Pi \Pi_{x^n} \Pi_{x^n z^n} \rho_{x^n z^n} \Pi_{x^n z^n} \Pi_{x^n}
         \Pi
         + \frac{\id}{|\cB|^n}
         \right) \\
& \leq & (1- 5\sqrt{\epsilon})^{-1} \sum_{z^n \in \cZ^n}
         p_{Z^n}(z^n|x^n) \Pi \Pi_{x^n} \rho_{x^n z^n} \Pi_{x^n} \Pi + 
         \frac{\id}{|\cB|^n} \\
&   =  & (1- 5\sqrt{\epsilon})^{-1} \Pi \Pi_{x^n} \rho_{x^n} \Pi_{x^n} \Pi
         + \frac{\id}{|\cB|^n} \\
& \leq & (1- 5\sqrt{\epsilon})^{-1} 2^{-n(H(B|X)_\rho - c(6\delta))}  
         \Pi \Pi_{x^n} \Pi + \frac{\id}{|\cB|^n},
\end{eqnarray*}
where we used the fact that $\Pi_{x^n z^n}$ and $\rho_{x^n z^n}$ commute.
This shows that 
\[
\ellinfty{\rho'^{B^n}_{x^n}} \leq 
(1-5\sqrt{\epsilon})^{-1} 2^{-n(H(B|X)_\rho - c(6\delta))} + 
|\cB|^{-n} \leq
4 \cdot 2^{-n(H(B|X)_\rho - c(6\delta))},
\]
since $\epsilon < 1/64$.

Finally,
\begin{eqnarray*}
\rho'^{B^n}
&   =  & \sum_{x^n \in \cX^n} p_{X^n}(x^n) \rho'^{B^n}_{x^n} \\
& \leq & \sum_{x^n \in \cX^n} p_{X^n}(x^n) \left(
         (1- 5\sqrt{\epsilon})^{-1} 
         \Pi \Pi_{x^n} \rho_{x^n} \Pi_{x^n} \Pi
         + \frac{\id}{|\cB|^n}
         \right) \\
& \leq & (1- 5\sqrt{\epsilon})^{-1} \sum_{x^n \in \cX^n}
         p_{X^n}(x^n) \Pi \rho_{x^n} \Pi + 
         \frac{\id}{|\cB|^n} \\
&   =  & (1- 5\sqrt{\epsilon})^{-1} \Pi \rho^{\otimes n} \Pi
         + \frac{\id}{|\cB|^n} \\
& \leq & (1- 5\sqrt{\epsilon})^{-1} 2^{-n(H(B)_\rho - c(2\delta))}  \Pi 
         + \frac{\id}{|\cB|^n},
\end{eqnarray*}
where we used the fact that $\Pi_{x^n}$ and $\rho_{x^n}$ commute.
This shows that 
\[
\ellinfty{\rho'^{B^n}} \leq 
(1-5\sqrt{\epsilon})^{-1} 2^{-n(H(B)_\rho - c(2\delta))} + 
|\cB|^{-n} \leq
4 \cdot 2^{-n(H(B)_\rho - c(2\delta))},
\]
since $\epsilon < 1/64$.
The proof is complete.
\end{proof}

\subsection{Geometric properties about projectors}
We first recall the following `gentle measurement' lemma of 
Winter~\cite{winter:strongconverse,ogawanagaoka:gentle}. The
lemma is typically used to assert that if a quantum state has a
high probability of succeeding on a measurement outcome, then the
collapsed state after the outcome has occured is close to the original
state.
\begin{fact}
\label{fact:gentle}
Let $M$ be a positive operator such that $M \leq \id$. 
Let $0 \leq \epsilon < 1$. Let $\rho$ be a density matrix. Suppose
$\Tr [M \rho M] \geq 1 - \epsilon$. Then,
$
\ellone{\rho - M \rho M} \leq 2 \sqrt{\epsilon}.
$
\end{fact}

The following geometric property is crucial to our construction of
a decoding strategy consisting of a sequence of orthogonal projectors.
It says that if we apply a sequence of orthogonal 
projectors to a pure state $\ket{\psi}$ such that $\ket{\psi}$ has a 
large projection onto
each projector, then the vector obtained at the end of the sequence is
close to the original vector $\ket{\psi}$. The advantage with respect
to Winter's `gentle measurement' lemma (Fact~\ref{fact:gentle}) 
is that in the gentle measurement lemma square roots of
the error probabilities under the operations add up, whereas in our 
lemma we exploit the fact that each operation is an orthogonal projection
to obtain that the error probabilities add up. 
The disadvantage of our lemma with respect to the gentle measurement lemma
is that our lemma guarantees closesness of the final state to
the initial state only for pure states. 
Nevertheless, we shall see by its corollary  Lemma~\ref{lem:seq} below
that our lemma can still be used to lower bound 
the success probability of a sequence of
orthogonal projectors applied to a mixed initial state. 
Our lemma allows us to avoid dependencies between the decoding 
operations
that were there in earlier sequential decoding strategies e.g. in 
Winter's strategy~\cite{winter:strongconverse}. It allows us 
to lower bound the decoding 
success probability
in a manner very similar to the classical setting. 
\begin{lemma}
\label{lem:key}
Let $v$ be a vector in a Hilbert space and $\Pi'_1, \ldots, \Pi'_k$ be
orthogonal projectors. Let $\Pi_i := \id - \Pi'_i$ be the projectors
onto the orthogonal subspaces. Then,
\[
\elltwo{v - \Pi'_k \cdots \Pi'_1 v}^2 \leq \sum_{i=1}^k \elltwo{\Pi_i v}^2.
\]
\end{lemma}
\begin{proof}
We prove this by induction on $k$. The null base case corresponding to
$k = 0$ is trivial. For the induction step, suppose
\[
\elltwo{v - \Pi'_{k-1} \cdots \Pi'_1 v}^2 \leq 
\sum_{i=1}^{k-1} \elltwo{\Pi_i v}^2.
\]
Then,
\begin{eqnarray*}
\elltwo{v - \Pi'_k \Pi'_{k-1} \cdots \Pi'_1 v}^2
&   =  & \elltwo{v - \Pi'_k v}^2 + 
         \elltwo{\Pi'_k v - \Pi'_k \Pi'_{k-1} \cdots \Pi'_1 v}^2 \\
& \leq & \elltwo{\Pi_k v}^2 + \elltwo{v -  \Pi'_{k-1} \cdots \Pi'_1 v}^2 \\
& \leq & \sum_{i=1}^k \elltwo{\Pi_i v}^2.
\end{eqnarray*}
The equality above uses the fact that the vector $(v - \Pi'_k v)$
is orthogonal to the support of $\Pi'_k$ since $\Pi'_k$ is an
orthogonal projector. In particular, $(v - \Pi'_k v)$ is orthogonal
to the vector $(\Pi'_k v - \Pi'_k \Pi'_{k-1} \cdots \Pi'_1 v)$. The first
inequality above follows from the fact that an orthogonal projector
like $\Pi'_k$ cannot increase the length of a vector. The proof is
complete.
\end{proof}

We can now prove the following corollary of Lemma~\ref{lem:key}. This
lemma is what we actually use in the error analysis of the sequential
decoder.
\begin{lemma}
\label{lem:seq}
Let $\rho$ be a positive operator on a Hilbert space and 
$\Tr \rho \leq 1$. Using the notation
of Lemma~\ref{lem:key}, 
\[
\Tr [\Pi'_k \cdots \Pi'_1 \rho \Pi'_1 \cdots \Pi'_k] \geq 
\Tr \rho - 2 \sqrt{\sum_{i=1}^k \Tr [\rho \Pi_i]}.
\]
\end{lemma}
\begin{proof}
By expressing $\rho$ in its eigenbasis and using the concavity of
the square root function, we can assume without loss
of generality for the purpose of this proof that $\rho$ has rank one i.e.
$\rho = \ketbra{v}$ for some vector $v$, $\elltwo{v} \leq 1$. 
That is, we just have to show
\[
\elltwo{\Pi'_k \cdots \Pi'_1 v}^2  \geq \elltwo{v}^2 - 
2 \sqrt{\sum_{i=1}^k \elltwo{\Pi_i v}^2}.
\]
By Lemma~\ref{lem:key},
\begin{eqnarray*}
\sqrt{\sum_{i=1}^k \elltwo{\Pi_i v}^2}
& \geq & \elltwo{v - \Pi'_k \cdots \Pi'_1 v} \\
& \geq & \elltwo{v} - \elltwo{\Pi'_k \cdots \Pi'_1 v},
\end{eqnarray*}
leading to
\begin{eqnarray*}
\elltwo{\Pi'_k \cdots \Pi'_1 v}^2  
& \geq & \elltwo{v}^2 - 2 \elltwo{v} \sqrt{\sum_{i=1}^k \elltwo{\Pi_i v}^2}
         + \sum_{i=1}^k \elltwo{\Pi_i v}^2 \\
& \geq & \elltwo{v}^2 - 2 \sqrt{\sum_{i=1}^k \elltwo{\Pi_i v}^2}.
\end{eqnarray*}
This completes the proof.
\end{proof}

Lemma~\ref{lem:seq} suggests the following natural sequential decoding 
paradigm for a channel. Let $\Pi_1, \ldots, \Pi_k, \ldots$ be 
orthogonal projectors. Intuitively, if the $i$th message were sent,
the output state $\rho_i$ of the channel would have a large projection
onto $\Pi_i$ i.e. $\Tr \Pi_i \rho_i \approx 1$. 
Supppose now that the $k$th message is sent resulting in output
state $\rho_k$ of the channel. The decoder tries to project
$\rho_k$ onto $\Pi_1$. If he fails, he then tries to project the collapsed
state onto $\Pi_2$. The decoder continues in this fashion. If he succeeds
in projecting onto $\Pi_i$ for some $i$, he declares $i$ to be his 
guess for the sent
message. If he fails to project onto any $\Pi_i$, he declares failure.
Using Lemma~\ref{lem:seq}, we can upper bound the decoding error 
probability $\err(k)$ of this strategy by 
\[
\err(k) \leq 2 \sqrt{\sum_{i = 1}^{k-1} \Tr [\rho_k \Pi_i] + 1
                     - \Tr [\rho_k \Pi_k]}.
\]
Our quantum sequential decoding paradigm is very analogous 
to the natural classical sequential decoding paradigm, and gives low
decoding error probability for any set of messages $\{1, 2, \ldots\}$
with channel output states $\{\rho_1, \rho_2, \ldots\}$ if we can
find, for example, 
a set of decoding orthogonal projectors $\{\Pi_1, \Pi_2, \ldots\}$
with the property that for each $k$, $\Tr [\rho_k \Pi_k] \approx 1$ and
$\sum_{i: i \neq k} \Tr [\rho_k \Pi_i] \approx 0$. In contrast,
the sequential decoding approach of Winter~\cite{winter:strongconverse}
requires extra conditions on the decoding projectors that involve
requiring a sequential ordering of the messages, which leads to 
difficulties in applying it to channels with multiple independent senders.
The sequential decoding approach of Giovannetti, Lloyd and 
Maccone~\cite{lloyd:seq} is heavily tied to the setting of the cq-channel
and requires independent choice of codewords for different
messages in the error analysis, which again leads to problems in applying
it to channels with multiple independent senders.
Moreover, they also need to use the projector onto the typical subspace
of the average channel state is also required besides the conditionally
typical projectors for each message.

Note that a similar upper bound on the decoding error probability can be 
obtained
for a decoder using a pretty-good-measurement (PGM) strategy, also
known as the square-root measurement strategy. Consider the 
PGM defined by the orthogonal projectors $\Pi_1, \ldots, \Pi_k, \ldots$
as follows:
\[
\Sigma := \sum_i \Pi_i, ~~~
\Upsilon_i := \Sigma^{-1/2} \Pi_i \Sigma^{-1/2},
\]
where $\Sigma^{-1/2}$ is defined in the natural fashion on the support
of $\Sigma$ and zero on the orthogonal complement.
Suppose the $k$th message is sent. The probability $\err'(k)$ of wrongly
decoding the message is $\Tr [(\id - \Upsilon_k) \rho_k]$, which
by the Hayashi-Nagaoka inequality~\cite{hayashinagaoka} is upper bounded
by
\[
\err'(k) \leq 2 (1 - \Tr [\rho_k \Pi_k]) + 
              4 \sum_{i: i \neq k} \Tr [\rho_k \Pi_i].
\]
In this sense, the sequential decoding strategy suggested by
Lemma~\ref{lem:seq} cannot give anything fundamentally 
different from known decoder construction techniques. 
The advantage of our sequential approach is that it gives a potentially
more efficient decoder than the pretty good measurement.
Also, the sequential decoding strategy is closer to the natural
classical intuition and arguably sheds more light on the interaction 
between the various projectors $\Pi_i$ in the decoding process.

In order to construct a simultaneous sequential decoder for the cq-MAC,
we will need the following geometric fact about how a pair of subspaces 
interact. This fact lies at the heart of several beautiful results in
quantum complexity theory \cite{szegedy:quantumwalk, watrous:qma, 
watrous:zk}.
\begin{fact}
\label{fact:chordal}
Let $A$, $B$ be subspaces of a Hilbert space $\cH$. Then there is a
decomposition of $\cH$ as an orthogonal direct sum of the 
following types of spaces:
\begin{enumerate}
\item One dimensional spaces orthogonal to both $A$ and $B$,
\item One dimensional spaces contained in both $A$ and $B$,
\item One dimensional spaces contained in $A$ but not in $B$,
\item One dimensional spaces contained in $B$ but not in $A$,
\item Two dimensional spaces intersecting $A$, $B$ each in one 
dimensional spaces.
\end{enumerate}
Moreover, the one dimensional spaces in (2) and (3) above and the one 
dimensional
intersections with $A$ of the two dimensional spaces in (5) above form
an orthonormal basis of $A$. A similar statement holds for $B$.
\end{fact}

Using Fact~\ref{fact:chordal}, we can prove the following lemma. This
lemma will be useful in designing a sequential simultaneous decoder for a
ccq-MAC.
\begin{lemma}
\label{lem:chordal}
Let $A'$, $B'$ be subspaces of a Hilbert space $\cH$. Let $\Pi_{A'}$, 
$\Pi_{B'}$ 
denote the orthogonal projectors onto $A'$, $B'$. 
Let $0 \leq \epsilon < 1$.
Then there exists a (possibly null) subspace $A$ of $A'$ 
such that for all vectors $v \in A$, 
$
\elltwo{\Pi_{B'} v}^2 \geq (1 - \sqrt{\epsilon}) \elltwo{v}^2.
$
Let $B$ be the (possibly null) subspace of $B'$ spanned by the vectors 
$\Pi_{B'} v$, 
$v \in A$. Let $\Pi_B$ be the orthogonal projector onto $B$. Then,
$\Pi_B \leq (1 -  \sqrt{\epsilon})^{-1} \Pi_{B'} \Pi_{A'} \Pi_{B'}$.
Suppose furthermore there is a positive matrix $\rho$, $\Tr \rho \leq 1$
with support in $A'$ such that 
$\Tr [\rho \Pi_{B'}] \geq 1 - \epsilon$. Then $A$ is a
non-null subspace and 
$\Tr [\rho \Pi_B] \geq 1 - 2 \sqrt{\epsilon}$.
\end{lemma}
\begin{proof}
Take the orthogonal direct sum decomposition of $\cH$ with respect
to the subspaces $A'$, $B'$ promised by Fact~\ref{fact:chordal}. 
Consider the basis $\{\ket{a'}\}_{a'}$ of $A'$ 
given by Fact~\ref{fact:chordal}. The important thing about this basis
is that the vectors $\Pi_{B'} \ket{a'}$ are orthogonal (some of them could
be zero). 
Let $\{\ket{a}\}_a$ range over the subset of these basis vectors 
satisfying the property that 
$\elltwo{\Pi_{B'} \ket{a}}^2 \geq 1 - \sqrt{\epsilon}$. Define $A$ to
be the subspace spanned by these basis vectors. Note that $A$ could be
the null space if there are no such basis vectors. Since
the vectors $\Pi_{B'} \ket{a}$ are orthogonal, we have that for
any vector $v \in A$, 
$
\elltwo{\Pi_{B'} v}^2 \geq (1 - \sqrt{\epsilon}) \elltwo{v}^2.
$
Observe that $B$ is spanned by the non-zero orthogonal vectors 
$\Pi_{B'} \ket{a}$. It is now easy to see that
\[
\Pi_B \leq (1 - \sqrt{\epsilon})^{-1} \Pi_{B'} \Pi_A \Pi_{B'} \leq
(1 - \sqrt{\epsilon})^{-1} \Pi_{B'} \Pi_{A'} \Pi_{B'}.
\]

Suppose now we have a positive matrix
$\rho$ with the desired properties. Let 
$\rho = \sum_i q_i \ketbra{\alpha_i}$ be a diagonalisation of $\rho$
into its eigenbasis. 
Therefore $\Tr [\rho \Pi_{B'}]$ can be described
by the success probability of the following random process:
\begin{enumerate}
\item With probability $1 - \Tr \rho$, abort;

\item Choose $\ket{\alpha_i}$ with probability $\frac{q_i}{\Tr \rho}$;

\item Choose $\ket{a'}$ with probability $|\braket{a'}{\alpha_i}|^2$;

\item Try to project $\ket{a'}$ onto $B'$, this succeeds with probability
      $\elltwo{\Pi_{B'} \ket{a'}}^2$.
\end{enumerate}
Since $\Tr [\rho \Pi_{B'}] \geq 1 - \epsilon$, with probability at
least $1 - \sqrt{\epsilon}$ over steps (1), (2) and (3), we get a basis
vector $\ket{a'}$ such that 
$\elltwo{\Pi_{B'} \ket{a'}}^2 \geq 1 - \sqrt{\epsilon}$. Thus, with
probability at least $1 - \sqrt{\epsilon}$ over steps (1), (2) and (3), 
we get a basis vector of subspace $A$. In particular, $A$ is a non-null
subspace and
$\Tr [\rho \Pi_A] \geq 1 - \sqrt{\epsilon}$, which in turn implies that
$\Tr [\rho \Pi_B] \geq (1 - \sqrt{\epsilon})^2 \geq 
 1 - 2 \sqrt{\epsilon}$. 
\end{proof}

\section{Sequential decoding for single sender single receiver channel}
\label{sec:cq}
Let $\chan: x \mapsto \rho_x$ be a channel which takes a classical input 
$x$ and outputs a quantum state $\rho_x$. We use 
$\rho^{X B} := \sum_{x \in \cX} p_X(x) \ketbra{x} \otimes \rho_x$ to 
denote
the joint state of the system $X B$. We will now show how to transmit
classical information over $\chan$ at any rate $R$ strictly less
than the mutual information $I(X : B)_\rho$.

Define $\rho := \sum_{x \in \cX} p_X(x) \rho_x$ to be the average
density matrix of the output system $B$. 
Let $0 < \epsilon, \delta < 1/64$. Define 
$p_{\mathrm{min}} := \min_{x \in \cX, p_X(x) > 0} p_X(x)$,
$q_{\mathrm{min}} := \min_{x \in \cX, p_X(x) > 0} 
 q_{\mathrm{min}}(\rho_x)$.
Let $n \geq 4 \delta^{-2} p_{\mathrm{min}}^{-1} q_{\mathrm{min}}^{-1}
     \log(|\cB| |\cX| / \epsilon)$.
Define $c(\delta) := \delta \log |\cB||\cX| - \delta \log \delta$. 
For $x^n \in T^{X^n}_\delta$, let $\Pi_{x^n}$ denote the conditionally
typical projector $\Pi^{\rho_{x^n}}_\delta$. 
For $x^n \not \in T^{X^n}_\delta$, define $\Pi_{x^n} := 0$. 
Suppose the sender wants
to transmit at rate $R$ using $n$ independent
uses of the channel $\chan$. Let his messages be denoted by $m$, 
$1 \leq m \leq 2^{n R}$. The sender uses the random encoding
strategy of Figure~\ref{fig:cqencoding} to choose a codebook $\codebook$.
\begin{figure}[!ht]
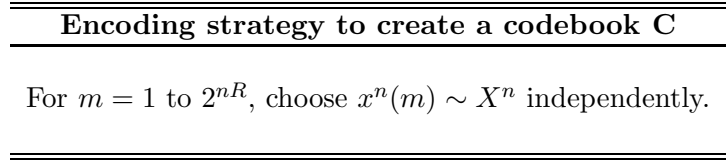

\begin{center}
\begin{tabular}{c}
\hline \hline 
{\bf Encoding strategy to create a codebook C} \\
\hline \\
For $m = 1 \mbox{ to } 2^{nR}$, choose $x^n(m) \sim X^n$ independently.
\\ \\
\hline \hline
\end{tabular}
\end{center}
\caption{Encoding strategy for the cq-channel.}
\label{fig:cqencoding}
\end{figure}

For the codebook $\codebook$, the receiver uses the sequential
decoding strategy of Figure~\ref{fig:cqdecoding}.
As will become clear later, the ordering of messages
is not very important but we impose the normal ordering while describing 
the decoding algorithm for clarity of exposition.
An interesting feature about our decoder is that it only uses projectors
onto conditionally typical subspaces of candidate messages. In particular,
it does not use any projection onto the typical subspace of the 
average message, unlike Giovannetti, Lloyd and Maccone's sequential
decoder~\cite{lloyd:seq} and many other decoders for the cq-channel
described earlier in the literature. The only other example of a
decoder for the cq-channel that solely uses projectors onto 
conditionally typical subspaces of candidate messages that we are aware
of is a pretty good measurement based decoder of 
Wilde~\cite[Exercises~15.5.4 and 19.3.5]{wilde:book}.
\begin{figure}[!ht]
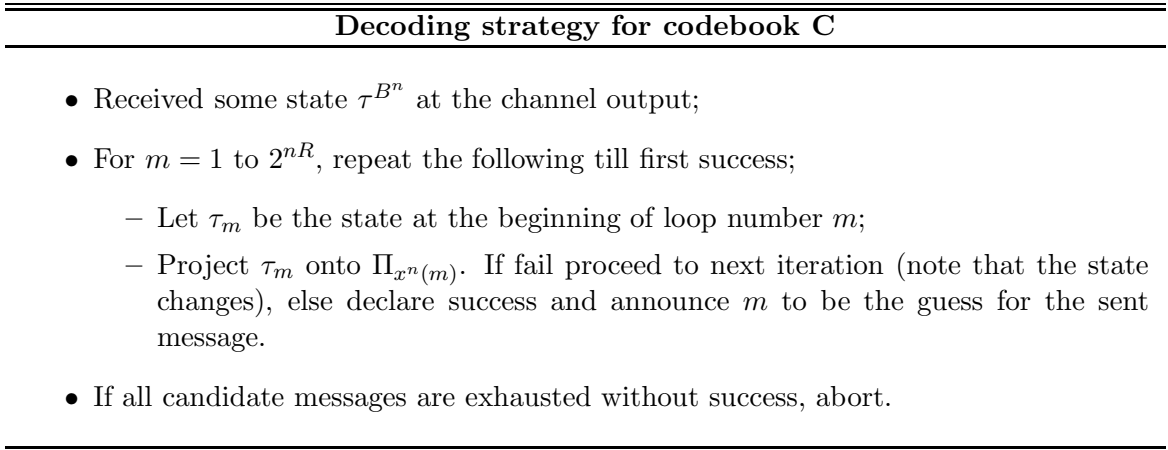

\begin{center}
\begin{tabular}{c}
\hline \hline 
{\bf Decoding strategy for codebook C} \\
\hline \\
\begin{minipage}[t]{15cm}
\begin{itemize}
\item Received some state $\tau^{B^n}$ at the channel output;

\item For $m = 1 \mbox{ to } 2^{nR}$, repeat the following till first
      success; 
      \begin{itemize}
      \item Let $\tau_m$ be the state at the beginning of loop number $m$;
     
      \item Project $\tau_m$ onto $\Pi_{x^n(m)}$. If fail proceed to next
            iteration (note that the state changes), else
            declare success and announce $m$ to be the guess for the sent
            message. 
      \end{itemize}

\item If all candidate messages are exhausted without success, abort.
\end{itemize}
\end{minipage} \\ \\
\hline \hline
\end{tabular}
\end{center}
\caption{Decoding strategy for the cq-channel.}
\label{fig:cqdecoding}
\end{figure}

In order to analyse the success probability of decoding a message
correctly, we first imagine that we are given a slightly 
different channel $\chan'$ working on alphabet $\cX^n$ with
output Hilbert space $\cB^{\otimes n}$.  It is defined as
$\chan': x^n \mapsto \rho'_{x^n}$ where
the states $\{\rho'^{B^n}_{x^n}\}_{x^n \in \cX^n}$
are as provided by Lemma~\ref{lem:asymsmooth} given state 
$\rho^{XB}$. Define the averaged density matrix
$
\rho' := \sum_{x^n \in \cX^n} p_{X^n}(x^n) \rho'_{x^n}.
$
Let $\rho'^{X^n B^n}$ be as defined by Lemma~\ref{lem:asymsmooth}. 
The encoding and decoding strategies used continue to be exactly as in 
Figures~\ref{fig:cqencoding} and \ref{fig:cqdecoding}, even though
the channel we are now working with is a single copy of $\chan'$.

For $1 \leq m \leq 2^{nR}$, define $\hPi_{x^n(m)} := \id - \Pi_{x^n(m)}$
to be the orthogonal projector corresponding to failing to project onto
$ \Pi_{x^n(m)}$.
The success probability of decoding message $m$ correctly for a single
copy of channel $\chan'$ is
\[
\Tr [\Pi_{x^n(m)} \hPi_{x^n(m-1)} \cdots \hPi_{x^n(1)} \rho'_{x^n(m)}
     \hPi_{x^n(1)} \cdots \hPi_{x^n(m-1)} \Pi_{x^n(m)}].
\]
By Lemma~\ref{lem:seq}, it is lower bounded by
\begin{eqnarray*}
\lefteqn{
1 - 2 
\sqrt{\sum_{i=1}^{m-1} \Tr [\Pi_{x^n(i)} \rho'_{x^n(m)}] +
      \Tr [\hPi_{x^n(m)} \rho'_{x^n(m)}]\;} 
}\\
& \geq & 1 - 2 
\sqrt{\sum_{i: i \neq m} \Tr [\Pi_{x^n(i)} \rho'_{x^n(m)}] + 1 -
      \Tr [\Pi_{x^n(m)} \rho'_{x^n(m)}]}.
\end{eqnarray*}
Intuitively, the last message has the largest decoding error.

We shall use the symbol $\E[\cdot]$ to denote expectation 
over the choice of the random codebook $\codebook$. 
We shall also use notation like $\E_{x^n}[\rho_{x^n}]$ to denote
$\sum_{x^n \in \cX^n} p_{X^n}(x^n) \rho_{x^n}$. 
The expected error 
probability $\E[\err(m)]$ of decoding message $m$ is now at most
\begin{eqnarray*}
\lefteqn{\E[\err(m)]} \\
& \leq & 2 \cdot \E\left[\sqrt{\sum_{i: i \neq m} 
                         \Tr [\Pi_{x^n(i)} \rho'_{x^n(m)}] + 1 -
                         \Tr [\Pi_{x^n(m)} \rho'_{x^n(m)}]}\right] \\
& \leq & 2 \sqrt{\sum_{i: i \neq m} 
                 \Tr [\E[\Pi_{x^n(i)} \rho'_{x^n(m)}]] + 1 -
                 \E[\Tr [\Pi_{x^n(m)} \rho'_{x^n(m)}]]} \\
&   =  & 2 \sqrt{(2^{nR} - 1)
                 \Tr [(\E_{x^n}[\Pi_{x^n}]) 
                      (\E_{x^n}[\rho'_{x^n}])] + 1 -
                 \E_{x^n}[\Tr [\Pi_{x^n} \rho'_{x^n}]]} \\
& \leq & 2 \sqrt{2^{nR} \cdot
                 \E_{x^n}[\Tr [\Pi_{x^n} \rho']] + 1 -
                 \E_{x^n}[\Tr [\Pi_{x^n} \rho'_{x^n}]]}. 
\end{eqnarray*}
The second inequality follows from concavity of square root. The 
equality follows from the independent and identical choice of codewords 
for a pair of different messages. 

We now bound the terms on the right
hand side above. For $x^n \in T^{X^n}_\delta$,
\begin{equation}
\label{eq:ccq1}
\begin{array}{c c l}
\Tr [\Pi_{x^n} \rho'] 
& \leq & \ellinfty{\rho'} \cdot \rank{\Pi_{x^n}} \\
& \leq & 4 \cdot 2^{-n(H(B)_\rho - c(2\delta))} \cdot 
         2^{n(H(B|X)_\rho+c(\delta))} \\
& \leq & 4 \cdot 2^{-n(I(X:B)_\rho - 2c(2\delta))},
\end{array}
\end{equation}
where we use Fact~\ref{fact:condtypicalspace} and 
Lemma~\ref{lem:asymsmooth}.
If $x^n \not \in T^{X^n}_\delta$,
$\Pi_{x^n} = 0$ and the upper bound on the trace above holds trivially.
Again, for $x^n \in T^{X^n}_\delta$,
\begin{equation}
\label{eq:ccq2}
\Tr [\Pi_{x^n} \rho'_{x^n}] 
  \geq   \Tr [\Pi_{x^n} \rho_{x^n}] - \ellone{\rho_{x^n} - \rho'_{x^n}}
  \geq   (1-\epsilon) - 11\sqrt{\epsilon} 
  \geq   1 - 12\sqrt{\epsilon},
\end{equation}
where we used Fact~\ref{fact:condtypicalspace} and 
Lemma~\ref{lem:asymsmooth}. This implies that
\[
\E_{x^n}[\Tr [\Pi_{x^n} \rho'_{x^n}]]
  \geq   (1-\epsilon) (1 - 12\sqrt{\epsilon})
  \geq   1 - 13\sqrt{\epsilon},
\]
by Fact~\ref{fact:typicalset}.
Putting all this together, we get
\[
\E[\err(m)]
  \leq   2 \sqrt{
         2^{nR} \cdot 4 \cdot 2^{-n(I(X:B)_\rho - 2c(2\delta))} 
	 + 13 \sqrt{\epsilon}
         }.  
\]
Choosing a rate 
$R = I(X:B)_\rho - 4 c(2\delta)$ gives us
$\E[\err(m)] \leq 8 \epsilon^{1/4} + 4 \cdot 2^{-nc(2\delta)}$,
for which if we take $n = C \log (1/\epsilon) \delta^{-2}$ 
for a constant
$C$ depending only on the state $\rho^{XB}$, we get
$\E[\err(m)] \leq 8 \epsilon^{1/4} + 
  4 \cdot 2^{-C\delta^{-2} c(2\delta) \log(1/\epsilon)}$. 

Since we have shown that $\E[\err(m)]$ is small for all messages
$m$, the expected average error probability 
$
\E[\avgerr_{\chan'}] := 
\E[2^{-nR} \sum_{m} \err_{\chan'}(m)]
$ 
for channel $\chan'$ is also small.
We now indicate why the same is true for $n$ copies of the original
channel $\chan$ using the same encoding and decoding procedures 
described in Figures~\ref{fig:cqencoding} and \ref{fig:cqdecoding}.

Let $m$ range over messages in $[2^{nR}]$.
Fix a codebook $\codebook$. Let $\cD^{\codebook}$ denote the POVM
corresponding to the decoding process for codebook $\codebook$ with
POVM elements $\cD^{\codebook}_{m}$. Define
$\hcD^{\codebook}_{m} := \id - \cD^{\codebook}_{m}$. 
Let $(x^n)^{\codebook}(m)$ denote the
codeword for message $m$ under $\codebook$.
The average error probability of decoding under 
$\codebook$ for channels $\chan$, $\chan'$ is given by
\begin{eqnarray*}
\avgerr^{\codebook}_{\chan} 
& = & 2^{-nR} \sum_{m} \Tr [\hcD^{\codebook}_{m} 
      \rho_{(x^n)^{\codebook}(m)}], \\
\avgerr^{\codebook}_{\chan'} 
& = & 2^{-nR} \sum_{m} \Tr [\hcD^{\codebook}_{m} 
      \rho'_{(x^n)^{\codebook}(m)}].
\end{eqnarray*}
As before we use the symbol $\E[\cdot]$ to denote
expectation over the choice of the random codebook $\codebook$, 
and symbols like 
$\E_{x^n y^n}[\cdot]$ to denote 
$\sum_{x^n, y^n} p_{X^n}(x^n) p_{Y^n}(y^n) (\cdot)$.
The difference between the expected average error probability of 
decoding for channels $\chan$, $\chan'$ is given by
\begin{eqnarray*}
\lefteqn{
|\E[\avgerr^{\codebook}_{\chan}] - \E[\avgerr^{\codebook}_{\chan'}]| 
} \\
&   =  & \left|
         \E\left[
         2^{-nR} \sum_{m} 
         \Tr [\hcD^{\codebook}_{m} 
              \rho_{(x^n)^{\codebook}(m)}] -
         \Tr [\hcD^{\codebook}_{m} 
              \rho'_{(x^n)^{\codebook}(m)}]
         \right]
         \right|\\
& \leq & \E\left[
         2^{-nR} \sum_{m} 
         |\Tr [\hcD^{\codebook}_{m} 
               \rho_{(x^n)^{\codebook}(m)}] -
          \Tr [\hcD^{\codebook}_{m} 
               \rho'_{(x^n)^{\codebook}(m)}]|
         \right]\\
& \leq & \E\left[
         2^{-nR} \sum_{m} 
         \ellone{\rho_{(x^n)^{\codebook}(m)}-
                 \rho'_{(x^n)^{\codebook}(m)}}
         \right] \\
&   =  & 2^{-nR} \sum_{m} 
         \E\left[
         \ellone{\rho_{(x^n)^{\codebook}(m)}-
                 \rho'_{(x^n)^{\codebook}(m)}}
         \right] \\
&   =  & 2^{-nR} \sum_{m} 
         \E_{x^n}[\ellone{\rho_{x^n} - \rho'_{x^n}}] 
\;  = \; \ellone{\rho^{X^n B^n} - \rho'^{X^n B^n}} \\
& \leq & 13\sqrt{\epsilon}.
\end{eqnarray*}
The last equality follows from the definitions of 
$\rho^{X^n B^n}$ and $\rho'^{X^n B^n}$. The last inequality
follows from Lemma~\ref{lem:asymsmooth}.

We have thus shown that with the encoding and decoding procedures
of Figures~\ref{fig:cqencoding} and \ref{fig:cqdecoding}, 
choosing a rate $R = I(X:B)_\rho - 4 c(2\delta)$ 
gives, for $n$ copies of the channel $\chan$, an expected average error
probability of at most
\[
\E[\avgerr^{\codebook}_{\chan}] \leq 
8\epsilon^{1/4} + 4 \cdot 2^{-C\delta^{-2} c(2\delta) \log(1/\epsilon)}
+ 13\sqrt{\epsilon} \leq 
21\epsilon^{1/4} + 4 \cdot 2^{-C\delta^{-2} c(2\delta) \log(1/\epsilon)},
\]
where $C$ is the same constant as above, depending only on the
channel $\chan$. This quantity can be made
arbitrarily small by choosing $\epsilon$ and $\delta$ appropriately.
We can also simultaneously make $c(2\delta)$ arbitrarily small so that
any rate $R$ strictly less than $I(X:B)_\rho$ can be achieved.
We have thus shown the following theorem.
\begin{theorem}
For a single sender single receiver channel $\chan$ with classical input 
and quantum output,
the random encoding and sequential decoding strategies described in
Figures~\ref{fig:cqencoding} and \ref{fig:cqdecoding} can achieve any
rate strictly less than the mutual information with arbitrary small
expected average probability of error in the asymptotic limit
of many independent uses of the channel.
\label{thm:cq}
\end{theorem}

\paragraph{A proof without using asymptotic smoothing:} It is possible 
to construct slightly more complicated
sequential decoders and improve the upper bound on the decoding error
probability that we can prove. For example, the decoder can first try
to project the received state onto the typical projector 
$\Pi := \Pi^{\rho^{\otimes n}}_{2\delta}$. If the attempt fails, then the
decoder aborts. If the attempt succeeds, then the decoder tries projecting
onto $\Pi_{x^n(m)}$ for $m = 1, \ldots, 2^{nR}$ as in 
Figure~\ref{fig:cqdecoding}. A very similar proof as above, using 
Lemma~\ref{lem:seq} with $\rho = \Pi \rho_{x^n(m)} \Pi$, shows
that the expected average decoding error probability is small. This
proof does not use asymptotic smoothing i.e. Lemma~\ref{lem:asymsmooth}.
However, the resulting upper bound on the error probability that we
can prove is not significantly better than the upper bound we prove above.
In the interest of keeping the decoder construction as simple
as possible, we have provided the proof above. 

\section{Jointly typical decoding for a multiple access channel}
\label{sec:MAC}

\subsection{Two sender MAC}
\label{subsec:ccq}
Let $\chan: (x,y) \mapsto \rho_{xy}$ be a channel that takes two classical 
inputs $x$ and $y$ and outputs a quantum state $\rho_{xy}$ in Hilbert
space $\cB$. We use
\begin{equation}
\rho^{XYB} := \sum_{(x,y) \in \cX \times \cY} p_X(x) p_Y(y)
\ketbra{x,y} \otimes \rho_{xy}
\label{eq:rhoXYB}
\end{equation}
to denote the joint state of the system $X Y B$. We will show
how to achieve the following standard inner bound for this ccq-MAC using a
{\em jointly typical} decoder performing a sequence of orthogonal
projections.
\begin{equation}
R_1 < I(X : B | Y)_\rho, ~~~~~
R_2 < I(Y : B | X)_\rho, ~~~~~
R_1 + R_2 < I(XY : B)_\rho.
\label{eq:ccq}
\end{equation}

Define the following averaged density matrices:
\[
\rho_x := \sum_{y \in \cY} p_Y(y) \rho_{xy}, ~~~~~
\rho_y := \sum_{x \in \cX} p_X(x) \rho_{xy}, ~~~~~
\rho   := \sum_{y \in \cY} p_Y(y) \rho_y.
\]
We can now define the $n$-fold tensor product matrices $\rho_{x^n y^n}$, 
$\rho_{x^n}$, $\rho_{y^n}$ for sequences $x^n$, $y^n$ in the 
natural manner. The following equations follow easily.
\begin{eqnarray*}
\forall x^n \in \cX^n: \rho_{x^n} 
& = & \sum_{y^n \in \cY^n} p_{Y^n}(y^n) \rho_{x^n y^n}, \\
\forall y^n \in \cY^n: \rho_{y^n} 
& = & \sum_{x^n \in \cX^n} p_{X^n}(x^n) \rho_{x^n y^n}, \\
\rho 
& = & \sum_{x^n \in \cX^n} p_{X^n}(x^n) \rho_{x^n} \\
& = & \sum_{y^n \in \cY^n} p_{Y^n}(y^n) \rho_{y^n}
  =   \sum_{(x^n,y^n) \in \cX^n \times \cY^n} 
           p_{X^n}(x^n) p_{Y^n}(y^n) \rho_{x^n y^n}.
\end{eqnarray*}

Let $0 < \epsilon, \delta < 1/64$. Define 
\begin{eqnarray*}
p_{\mathrm{min}} 
& := & \min_{(x,y) \in \cX \times \cY, p_X(x) p_Y(y) > 0} p_X(x) p_Y(y),\\
q_{\mathrm{min}} 
& := & \min_{(x,y) \in \cX \times \cY, p_X(x) p_Y(y) > 0} 
       q_{\mathrm{min}}(\rho_{x,y}).
\end{eqnarray*}
Let $n \geq 4 \delta^{-2} p_{\mathrm{min}}^{-1} q_{\mathrm{min}}^{-1}
     \log(|\cB| |\cX| |\cY| / \epsilon)$.
Define $c(\delta) := \delta \log |\cB||\cX||\cY| - \delta \log \delta$. 
Suppose that $x^n y^n \in T^{(XY)^n}_\delta$. Note that this implies that
$x^n \in T^{X^n}_\delta$ and $y^n \in T^{Y^n}_\delta$; however, the 
converse is not true.
Let $\Pi_{x^n y^n}$, $\Pi_{x^n}$, $\Pi_{y^n}$ denote the conditionally
typical projectors $\Pi^{\rho_{x^n y^n}}_\delta$, 
$\Pi^{\rho_{x^n}}_{6\delta}$,
$\Pi^{\rho_{y^n}}_{6\delta}$.
We define a projector $\tPi_{x^n y^n}$ as follows.
Treat the supports of the conditionally typical projectors 
$\Pi_{x^n y^n}$,
$\Pi_{y^n}$ as the subspaces $A'$, $B'$ of Lemma~\ref{lem:chordal}, and 
obtain $\tPi_{x^n y^n}$ as the projector $\Pi_B$
promised by Lemma~\ref{lem:chordal}. If 
$x^n y^n \not \in T^{(XY)^n}_\delta$, define $\tPi_{x^n y^n} := 0$.
The decoding procedure shall use
the projectors $\tPi_{x^n y^n}$ instead of the conditionally
typical projectors $\Pi_{x^n y^n}$ one might naively expect.

Suppose senders 1, 2
want to transmit at rates $R_1$, $R_2$ using $n$ independent
uses of the channel $\chan$. Let their messages be denoted by $m_1$, $m_2$
respectively, where $1 \leq m_i \leq 2^{n R_i}$. The senders use the 
random encoding strategy of Figure~\ref{fig:ccqencoding}
 to choose a codebook $\codebook$.
\begin{figure}[!ht]
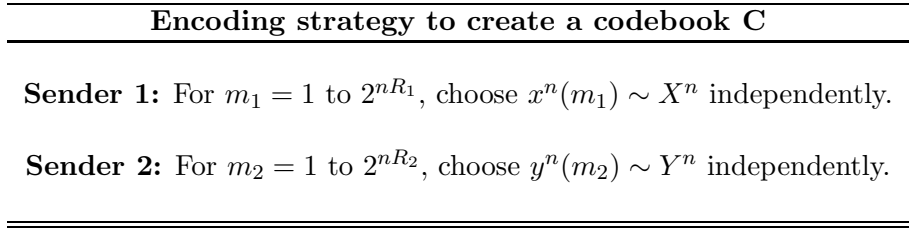

\begin{center}
\begin{tabular}{c}
\hline \hline 
{\bf Encoding strategy to create a codebook C} \\
\hline \\
{\bf Sender 1:}
For $m_1 = 1 \mbox{ to } 2^{nR_1}$, choose $x^n(m_1) \sim X^n$ 
independently.  \\ \\
{\bf Sender 2:}
For $m_2 = 1 \mbox{ to } 2^{nR_2}$, choose $y^n(m_2) \sim Y^n$ 
independently. \\ \\
\hline \hline
\end{tabular}
\end{center}
\caption{Encoding strategy for ccq-MAC.}
\label{fig:ccqencoding}
\end{figure}

For the codebook $\codebook$, the receiver uses the sequential
decoding strategy of Figure~\ref{fig:ccqdecoding}.
As will become clear later, the ordering of messages
is not very important but we impose the lexicographic ordering
looping over $m_2$ first and $m_1$ later, while describing the
decoding algorithm, for clarity of exposition.
\begin{figure}[!ht]
\begin{center}
\begin{tabular}{c}
\hline \hline 
{\bf Decoding strategy for codebook C} \\
\hline \\
\begin{minipage}[t]{15cm}
\begin{itemize}
\item Received some state $\tau^{B^n}$ at the channel output;

\item 
\begin{minipage}[t]{13.5cm}
\begin{tabbing}
For \= $m_1 = 1$ to  $2^{nR_1}$, \\
\> For \= $m_2 = 1$ to $2^{nR_2}$, \\
\> \> Repeat the following till first success; \\
\> \> \begin{minipage}[t]{12.5cm}
      \begin{itemize}
      \item Let $\tau_{m_1 m_2}$ be the state at the beginning of loop 
            number $(m_1, m_2)$;
     
      \item Project $\tau_{m_1 m_2}$ onto $\tPi_{x^n(m_1) y^n(m_2)}$.
            If fail proceed, else
            declare success and announce $(m_1, m_2)$ to be the guess for 
            the sent message pair.
      \end{itemize}
      \end{minipage}
\end{tabbing}
\end{minipage}

\item If all candidate message pairs are exhausted without success, abort.
\end{itemize}
\end{minipage} \\ \\
\hline \hline
\end{tabular}
\end{center}
\caption{Decoding strategy for ccq-MAC.}
\label{fig:ccqdecoding}
\end{figure}

In order to analyse the success probability of decoding a message
correctly, we first imagine that we are given a slightly 
different channel $\chan'$ working on alphabets $\cX^n$, $\cY^n$ with
output Hilbert space $\cB^{\otimes n}$.  It is defined as
$\chan': (x^n, y^n) \mapsto \rho'_{x^n y^n}$ where
the states $\{\rho'^{B^n}_{x^n y^n}\}_{(x^n, y^n) \in \cX^n \times \cY^n}$
are as provided by Lemma~\ref{lem:asymsmooth} given state 
$\rho^{XYB}$. Define the following averaged density matrices:
\[
\rho'_{x^n} := \sum_{y^n \in \cY^n} p_{Y^n}(y^n) \rho'_{x^n y^n}, ~~~~~
\rho'_{y^n} := \sum_{x^n \in \cX^n} p_{X^n}(x^n) \rho'_{x^n y^n}, ~~~~~
\rho'       := \sum_{y^n \in \cY^n} p_{Y^n}(y^n) \rho'_{y^n}.
\]
Let $\rho'^{X^n Y^n B^n}$ be as defined by Lemma~\ref{lem:asymsmooth}. 
The encoding and decoding strategies used continue to be exactly as in 
Figures~\ref{fig:ccqencoding} and \ref{fig:ccqdecoding}, even though
the channel we are now working with is a single copy of $\chan'$.

For $(m_1, m_2) \in [2^{nR_1}] \times [2^{n R_2}]$, define 
\[
\htPi_{x^n(m_1) y^n(m_2)} := 
\id - \tPi_{x^n(m_1) y^n(m_2)}
\]
to be the orthogonal projector corresponding to failing to project onto
$\tPi_{x^n(m_1) y^n(m_2)}$.
Let $\mm(m_1,m_2)$ denote the predecessor of $(m_1,m_2)$ in the 
lexicographic order if it exists. We shall use the notation
$(x^n,y^n)(\mm(m_1,m_2))$
to denote $(x^n(m'_1), y^n(m'_2))$ where $(m'_1, m'_2) := \mm(m_1,m_2)$.

The success probability of decoding $(m_1, m_2)$ correctly
for a single copy of channel $\chan'$ is
\begin{eqnarray*}
\Tr [\tPi_{x^n(m_1) y^n(m_2)} \htPi_{(x^n,y^n)(\mm(m_1,m_2))}
     \cdots \htPi_{x^n(1) y^n(1)} \; \rho'_{x^n(m_1) y^n(m_2)} \\
     \htPi_{x^n(1) y^n(1)} \cdots 
     \htPi_{(x^n,y^n)(\mm(m_1,m_2))} \tPi_{x^n(m_1) y^n(m_2)}
    ].
\end{eqnarray*}
By Lemma~\ref{lem:seq}, it is lower bounded by
\begin{eqnarray*}
\lefteqn{
1 - 2 \cdot 
    \sqrt{\sum_{(i,j)=(1,1)}^{\mm(m_1,m_2)} 
          \Tr [\tPi_{x^n(i) y^n(j)}  \rho'_{x^n(m_1) y^n(m_2)}] +
               \Tr [\htPi_{x^n(m_1) y^n(m_2)} \rho'_{x^n(m_1) y^n(m_2)}]
         } 
}\\
& \geq & 1 -  \\
&      & 2 \cdot
         \sqrt{\sum_{\stackrel{(i,j):}{(i,j) \neq (m_1,m_2)}}
               \Tr [\tPi_{x^n(i) y^n(j)} \rho'_{x^n(m_1) y^n(m_2)}]
               + 1 - 
               \Tr [\tPi_{x^n(m_1) y^n(m_2)} \rho'_{x^n(m_1) y^n(m_2)}]
              }\;.
\end{eqnarray*}

We shall use the symbol $\E[\cdot]$ to denote expectation 
over the choice of the random codebook $\codebook$.
We shall also use notation like 
$\E_{x^n y^n} [\rho'_{x^n y^n}] :=
 \sum_{x^n, y^n} p_{X^n}(x^n) p_{Y^n}(y^n) \rho'_{x^n y^n}$, and
$\E_{y^n}[\rho'_{x^n y^n}] := 
 \sum_{y^n} p_{Y^n}(y^n) \rho'_{x^n y^n}.
$
The expected error 
probability $\E[\err(m_1,m_2)]$ of decoding $(m_1,m_2)$ is now at 
most
\begin{eqnarray*}
\lefteqn{\E[\err(m_1,m_2)]} \\
& \leq & 2 \cdot 
         \E\left[
         \sqrt{\sum_{\stackrel{(i,j):}{(i,j) \neq (m_1,m_2)}}
               \Tr [\tPi_{x^n(i) y^n(j)} \rho'_{x^n(m_1) y^n(m_2)}]
               + 1 - 
               \Tr [\tPi_{x^n(m_1) y^n(m_2)} \rho'_{x^n(m_1) y^n(m_2)}]
              }
           \right] \\
& \leq & \!\!\!\!
         2 \cdot
         \!
         \left(
         \begin{array}{l l}
         \displaystyle{\sum_{i: i \neq m_1}}
         \Tr [\E[\tPi_{x^n(i) y^n(m_2)} \rho'_{x^n(m_1) y^n(m_2)}]] +
         \displaystyle{\sum_{j: j \neq m_2}}
         \Tr [\E[\tPi_{x^n(m_1) y^n(j)} \rho'_{x^n(m_1) y^n(m_2)}]] 
         + {} \\
         \displaystyle{\sum_{\stackrel{(i,j):}{i \neq m_1, j \neq m_2}}}
         \Tr [\E[\tPi_{x^n(i) y^n(j)} \rho'_{x^n(m_1) y^n(m_2)}]] 
         + 1 - 
         \E[\Tr [\tPi_{x^n(m_1) y^n(m_2)} \rho'_{x^n(m_1) y^n(m_2)}]]
         \end{array}
         \right)^{1/2} \\
&   =  & 2 \cdot
         \left(
         \begin{array}{l l}
         (2^{nR_1} - 1)
         \Tr [\E_{y^n}[(\E_{x^n}[\tPi_{x^n y^n}]) 
                       (\E_{x^n}[\rho'_{x^n y^n}])]] + {} \\
         (2^{nR_2} - 1)
         \Tr [\E_{x^n}[(\E_{y^n}[\tPi_{x^n y^n}]) 
                       (\E_{y^n}[\rho'_{x^n y^n}])]] + {} \\ 
         (2^{nR_1} - 1)(2^{n R_2} - 1)
         \Tr [(\E_{x^n y^n}[\tPi_{x^n y^n}])
              (\E_{x^n y^n}[\rho'_{x^n y^n}])] + {} \\
         1 -
         \E_{x^n y^n}[\Tr [\tPi_{x^n y^n} \rho'_{x^n y^n}]]
         \end{array}
         \right)^{1/2} \\
& \leq & 2 \cdot
         \left(
         \begin{array}{l l}
         2^{nR_1} \cdot
         \E_{x^n y^n}[\Tr [\tPi_{x^n y^n} \rho'_{y^n}]] + 
         2^{nR_2} \cdot
         \E_{x^n y^n}[\Tr [\tPi_{x^n y^n} \rho'_{x^n}]] + {} \\ 
         2^{n(R_1+R_2)} \cdot
         \E_{x^n y^n}[\Tr [\tPi_{x^n y^n} \rho']] + 1 -
         \E_{x^n y^n}[\Tr [\tPi_{x^n y^n} \rho'_{x^n y^n}]]
         \end{array}
         \right)^{1/2}.
\end{eqnarray*}
The second inequality follows from concavity of square root. The 
equality follows from the independent and identical choice of codewords 
for a pair of different messages. 

We now bound the terms on the right
hand side above. If $x^n y^n \in T^{(XY)^n}_\delta$, then
\begin{eqnarray*}
\Tr [\tPi_{x^n y^n} \rho'_{y^n}]
& \leq & (1-\sqrt{\epsilon})^{-1}
         \Tr [\Pi_{y^n} \Pi_{x^n y^n} \Pi_{y^n} \rho'_{y^n}] \\
& \leq & 2^{n(H(B|XY)_\rho + c(\delta))} (1-\sqrt{\epsilon})^{-1}
         \Tr [\Pi_{y^n} \rho_{x^n y^n} \Pi_{y^n} \rho'_{y^n}] 
\end{eqnarray*}
where we use the definition of $\tPi_{x^n y^n}$,
Lemma~\ref{lem:chordal} and Fact~\ref{fact:condtypicalspace}.
If $x^n y^n \not \in T^{(XY)^n}_\delta$,
then $\tPi_{x^n y^n} = 0$ and the upper bound on the trace above
holds trivially. Fix $y^n \in T^{Y^n}_\delta$. We get
\begin{eqnarray*}
\E_{x^n}[\Tr [\tPi_{x^n y^n} \rho'_{y^n}]]
& \leq & 2^{n(H(B|XY)_\rho + c(\delta))} (1-\sqrt{\epsilon})^{-1}
         \E_{x^n}[\Tr [\Pi_{y^n} \rho_{x^n y^n} \Pi_{y^n} \rho'_{y^n}]] \\
&   =  & 2^{n(H(B|XY)_\rho + c(\delta))} (1-\sqrt{\epsilon})^{-1}
         \Tr [\Pi_{y^n} (\E_{x^n}[\rho_{x^n y^n}]) \Pi_{y^n} \rho'_{y^n}]\\
&   =  & 2^{n(H(B|XY)_\rho + c(\delta))} (1-\sqrt{\epsilon})^{-1}
         \Tr [\Pi_{y^n} \rho_{y^n} \Pi_{y^n} \rho'_{y^n}] \\
& \leq & 2^{-n(H(B|Y)_\rho - c(6\delta))} \cdot 2^{n(H(B|XY)+c(\delta))}
         (1-\sqrt{\epsilon})^{-1}
         \Tr [\Pi_{y^n} \rho'_{y^n}] \\
& \leq & 2^{-n(I(X:B|Y)_\rho - 2c(6\delta))} (1-\sqrt{\epsilon})^{-1},
\end{eqnarray*}
where we use Fact~\ref{fact:condtypicalspace}. If 
$y^n \not \in T^{Y^n}_\delta$, then for all $x^n \in \cX^n$,
$x^n y^n \not \in T^{(XY)^n}_\delta$ and $\tPi_{x^n y^n} = 0$, so
the above upper bound on $\E_{x^n}[\Tr [\tPi_{x^n y^n} \rho'_{y^n}]]$
holds trivially.
We can now understand better the 
reason why the projectors $\tPi_{x^n y^n}$ are used in the decoding
process instead of the conditionally typical projectors 
$\Pi_{x^n y^n}$ one may naively expect. 
Suppose $x^n y^n \in T^{(XY)^n}_\delta$.
If we had used just 
$\Pi_{x^n y^n}$, we would not have been able to surround it by the
projectors $\Pi_{y^n}$ in the above chain of inequalities and we would
not have been able to prove the desired exponentially small upper bound.
If we had used the positive operator $\Pi_y \Pi_{x^n y^n} \Pi_y$ instead,
we would not have been able to apply Lemma~\ref{lem:seq} which is 
fundamental to our decoding strategy, since in general $\Pi_y$ does not
commute with $\Pi_{x^n y^n}$ and hence $\Pi_y \Pi_{x^n y^n} \Pi_y$
is not a projector. The operator $\tPi_{x^n y^n}$ mimics the
intuitive notion of projecting onto the intersection of the typical
spaces of $\rho_{x^n y^n}$ and $\rho_{y^n}$ 
in the commutative (classical) case. It is a bonafide orthogonal 
projector and smaller
(up to an error factor of $1-\sqrt{\epsilon}$) than the operator
$\Pi_{y^n} \Pi_{x^n y^n} \Pi_{y^n}$. Hence it can be fruitfully used in our
decoding strategy and its error analysis, and we can prove our desired 
exponentially small upper bound on the error probability.

Continuing we get, for $x^n y^n \in T^{(XY)^n}_\delta$,
\begin{eqnarray*}
\Tr [\tPi_{x^n y^n} \rho'_{x^n}]
& \leq & \ellinfty{\rho'_{x^n}} \cdot \Tr [\tPi_{x^n y^n}] 
\;\leq\; 4 \cdot 2^{-n(H(B|X)_\rho - c(6\delta))} \Tr [\Pi_{x^n y^n}] \\
& \leq & 4 \cdot 2^{-n(H(B|X)_\rho - c(6\delta))} \cdot
         2^{n(H(B|XY)_\rho + c(\delta))} 
\;\leq\; 4 \cdot 2^{-n(I(Y:B|X)_\rho - 2c(6\delta))},
\end{eqnarray*}
where we use Lemma~\ref{lem:asymsmooth}, 
Fact~\ref{fact:condtypicalspace} and the facts that 
$\rank{\tPi_{x^n y^n}} \leq \rank{\Pi_{x^n y^n}}$ and
$\rho'_{x^n} \leq \ellinfty{\rho'_{x^n}} \id$.
If $x^n y^n \not \in T^{(XY)^n}_\delta$, then $\tPi_{x^n y^n} = 0$
and the upper bound on the trace above holds trivially.

Similarly, for $x^n y^n \in T^{(XY)^n}_\delta$ we have
\begin{eqnarray*}
\Tr [\tPi_{x^n y^n} \rho']
& \leq & \ellinfty{\rho'} \cdot \Tr [\tPi_{x^n y^n}]
\;\leq\; 4 \cdot 2^{-n(H(B)_\rho - c(2\delta))} \cdot 
         \Tr [\Pi_{x^n y^n}]\\
& \leq & 4 \cdot 2^{-n(H(B)_\rho - c(2\delta))} \cdot
         2^{n(H(B|XY)_\rho + c(\delta))} 
\;\leq\; 4 \cdot 2^{-n(I(XY:B)_\rho - 2c(2\delta))},
\end{eqnarray*}
where we use Lemma~\ref{lem:asymsmooth}, Fact~\ref{fact:condtypicalspace}
and the facts that 
$\rank{\tPi_{x^n y^n}} \leq \rank{\Pi_{x^n y^n}}$ and
$\rho' \leq \ellinfty{\rho'} \id$.
If $x^n y^n \not \in T^{(XY)^n}_\delta$, then $\tPi_{x^n y^n} = 0$
and the upper bound on the trace above holds trivially. 

Finally, for $x^n y^n \in T^{(XY)^n}_\delta$,
\begin{eqnarray*}
\lefteqn{\Tr [\tPi_{x^n y^n} \rho'_{x^n y^n}]}\\
& \geq & \Tr [\tPi_{x^n y^n} \Pi_{x^n y^n} 
              \rho_{x^n y^n} \Pi_{x^n y^n}] - 
         \ellone{\rho_{x^n y^n} - 
                 \Pi_{x^n y^n} \rho_{x^n y^n} \Pi_{x^n y^n}} - 
         \ellone{\rho_{x^n y^n} - \rho'_{x^n y^n}} \\
& \geq & (1-4\sqrt{\epsilon}) - 2\epsilon - 
         11\sqrt{\epsilon}
\;\geq\; 1 - 17\sqrt{\epsilon},
\end{eqnarray*}
where we used Lemma~\ref{lem:asymsmooth}, 
Fact~\ref{fact:condtypicalspace} and the fact that $\Pi_{x^n y^n}$
commutes with $\rho_{x^n y^n}$.
For the
second inequality we also used Lemma~\ref{lem:chordal} to conclude that
$
\Tr [\tPi_{x^n y^n} (\Pi_{x^n y^n} \rho_{x^n y^n}
     \Pi_{x^n y^n})] \geq 1 - 4\sqrt{\epsilon},
$
since
\[
\Tr [\Pi_{y^n} (\Pi_{x^n y^n} \rho_{x^n y^n} \Pi_{x^n y^n})] \geq
\Tr [\Pi_{y^n} \rho_{x^n y^n}] - 
\ellone{\rho_{x^n y^n} - \Pi_{x^n y^n} \rho_{x^n y^n} \Pi_{x^n y^n}} \geq
1 - \epsilon - 2\epsilon = 1 - 3\epsilon,
\]
where we used Facts~\ref{fact:winter}, \ref{fact:condtypicalspace} and 
the fact that $\Pi_{x^n y^n}$ commutes with $\rho_{x^n y^n}$. 
We proceed to get
\[
\E_{x^n y^n}[\Tr [\tPi_{x^n y^n} \rho'_{x^n y^n}]] \geq
(1 - \epsilon) (1 - 17\sqrt{\epsilon}) \geq
1 - 18\sqrt{\epsilon},
\]
using Fact~\ref{fact:typicalset}.

Putting all this together, we get
\begin{eqnarray*}
\lefteqn{\E[\err(m_1,m_2)]} \\
& \leq & 2 \cdot
         \sqrt{
         \begin{array}{l l}
         2^{nR_1} \cdot 2^{-n(I(X:B|Y)_\rho - 2c(6\delta))} 
         (1-\sqrt{\epsilon})^{-1} +
         4 \cdot 2^{nR_2} \cdot 2^{-n(I(Y:B|X)_\rho - 2c(6\delta))} \\
         {} + 4 \cdot 2^{n(R_1+R_2)} \cdot 
              2^{-n(I(XY:B)_\rho - 2c(6\delta))}
         + 18\sqrt{\epsilon}
         \end{array}
         }.
\end{eqnarray*}
Choosing rate pairs satisfying the inequalities
\begin{equation}
R_1 \leq I(X:B|Y)_\rho - 4 c(6\delta),
R_2 \leq I(Y:B|X)_\rho - 4 c(6\delta),
R_1 + R_2 \leq I(XY:B)_\rho - 4 c(6\delta), 
\label{eq:ccqMACprime}
\end{equation}
and using the fact that $\epsilon < 1/64$
gives us
$\E[\err(m_1,m_2)] \leq 10\epsilon^{1/4} + 8 \cdot 2^{-nc(6\delta)}$,
for which if we take $n = C \log (1/\epsilon) \delta^{-2}$ 
for a constant
$C$ depending only on the state $\rho^{XYB}$, we get
$\E[\err(m_1,m_2)] \leq 10\epsilon^{1/4} + 
  8 \cdot 2^{-C\delta^{-2} c(6\delta) \log(1/\epsilon)}$.

Since we have shown that $\E[\err(m_1,m_2)]$ is small for all message pairs
$(m_1,m_2)$, the expected average error 
probability 
$
\E[\avgerr_{\chan'}] := 
\E[2^{-n(R_1+R_2)} \sum_{m_1,m_2} \err_{\chan'}(m_1,m_2)]
$ 
for channel $\chan'$ is also small.
We can now apply an argument similar to that of the proof of
Theorem~\ref{thm:cq} to show that the expected average error probability
is also small for $n$ copies of the original
channel $\chan$ for the same encoding and decoding procedures as
described in Figures~\ref{fig:ccqencoding} and \ref{fig:ccqdecoding}.
Doing so allows us to show that
\[
\E[\avgerr_{\chan}] \leq \E[\avgerr_{\chan'}] + 13\sqrt{\epsilon}
\leq 23\epsilon^{1/4} + 
8 \cdot 2^{-C\delta^{-2} c(6\delta) \log(1/\epsilon)},
\]
where $C$ is the same constant as above, depending only on the channel
$\chan$. This quantity can be made
arbitrarily small by choosing $\epsilon$ and $\delta$ appropriately.
We can also simultaneously make $c(6\delta)$ arbitrarily small so that
the rate pair $(R_1, R_2)$ approaches the boundary of the region
described in Equation~\ref{eq:ccq}.
We have thus shown the following theorem.
\begin{theorem}
For a two sender multiple access channel $\chan$ with classical input and
quantum output,
the random encoding and sequential decoding strategies described in
Figures~\ref{fig:ccqencoding} and \ref{fig:ccqdecoding} can 
achieve any rate in the region described by 
Equation~\ref{eq:ccq} with arbitrary small expected
average probability of error in the asymptotic limit
of many independent uses of the channel.
\label{thm:ccq}
\end{theorem}

\paragraph{A pretty good measurement decoder:} 
We can construct a jointly typical decoder for the ccq-MAC
using the pretty good measurement with
positive operators $\Pi_{y^n(m_2)} \Pi_{x^n(m_1) y^n(m_2)} 
\Pi_{y^n(m_2)}$. We analyse the decoding error as above, imagining
that the channel we are working with is a single copy of $\chan'$ and
then concluding that the decoding error for $n$ copies of the original
channel $\chan$ is small. The same decoding strategy was independently
discovered by Fawzi et al.~\cite{mcgill:qic}, and by Xu and 
Wilde~\cite{xuwilde:assisted} for the entanglement assisted case.

\subsection{Extension to MACs with three or more senders?}
We have constructed a jointly typical decoder for a two-sender MAC.
A natural question that arises is whether we can get jointly typical
decoders for MACs with three or more senders. However, our method 
runs into problems with three senders. Observe that our bounds on 
$R_2$ and $R_1 + R_2$ used that $\ellinfty{\rho'_{x^n}}$ and 
$\ellinfty{\rho'}$ are small, whereas our bound on 
$R_1$ came from the property that $\tPi_{x^n y^n}$ is roughly
less than $\Pi_{y^n} \Pi_{x^n y^n} \Pi_{y^n}$.
For a three sender MAC we can, for example, get bounds on $R_1$ and
$R_1 + R_2$ by constructing, 
using Lemma~\ref{lem:chordal} twice, an orthogonal projection 
$\tPi_{x^n y^n z^n}$ such
that it is roughly less than 
$\Pi_{z^n} \Pi_{y^n z^n} \Pi_{x^n y^n z^n} \Pi_{y^n z^n} \Pi_{z^n}$.
We can also get the desired bound on $R_2$, $R_2 + R_3$ and 
$R_1 + R_2 + R_3$ by doing asymptotic smoothing using 
Lemma~\ref{lem:asymsmooth} which gives us a state
$\rho'^{X^n Y^n Z^n B^n}$ close to $\rho^{X^n Y^n Z^n B^n}$ such that
$\ellinfty{\rho'_{x^n z^n}}$, $\ellinfty{\rho'_{x^n}}$ and
$\ellinfty{\rho'}$ are small.
However now if we want to get the desired bound on $R_3$
we would need $\ellinfty{\rho'_{x^n y^n}}$ to be small too.
We do not know how to prove the existence of such states
$\rho'^{X^n Y^n Z^n B^n}$. The
asymptotic smoothing methods of Section~\ref{subsec:asymsmooth} seem 
to fail to show this. 

If we assume that some suitable averaged output states commute, then
we can construct a jointly typical decoder for the three sender MAC.
For example, suppose that for all $x$, $y$, $z$, $\rho_{yz}$ commutes
with $\rho_{xy}$ for the same value of
$y$, and also $\rho_y$ commutes with $\rho_z$. Then we can
choose $\tPi_{x^n y^n z^n}$ such that it is roughly less than
$\Pi_{z^n} \Pi_{y^n z^n} \Pi_{x^n y^n z^n} \Pi_{y^n z^n} \Pi_{z^n}$
as well as 
$\Pi_{y^n} \Pi_{x^n y^n} \Pi_{x^n y^n z^n} \Pi_{x^n y^n} \Pi_{y^n}$ by
using Lemma~\ref{lem:chordal} twice 
with the intersections of supports of $\Pi_{x^n y^n}$ and
$\Pi_{y^n z^n}$, and $\Pi_{y^n}$ and $\Pi_{z^n}$.
This allows us to get the desired bounds on $R_1$, $R_1 + R_2$, 
$R_3$ and $R_1 + R_3$. The desired bounds on $R_2$, $R_2 + R_3$ and
$R_1 + R_2 + R_3$ can be obtained using asymptotic smoothing as
before.

Note that we can easily do successive cancellation by our
sequential decoding method and approach any vertex point of the rate region
of a MAC with any number of senders. For example, to approach the vertex
point $(I(X:B), I(Y:B|X), I(Z:B|XY))$ of the rate region of a 
three-sender MAC, we do successive cancellation using the 
decoding method for a cq-channel described in Section~\ref{sec:cq} earlier
as follows. Let the transmitted messages be $m_1$, $m_2$, $m_3$
where $m_i \in 2^{nR_i}$, $R_1 < I(X:B)$, $R_2 < I(Y:B|X)$,
$R_3 < I(Z:B|XY)$.
First decode $x^n(m_1)$ by applying a sequence of projectors
$\Pi_{x^n(i)}$, $i = 1, \ldots, m_1$, then decode $y^n(m_2)$ by applying
a sequence of projectors $\Pi_{x^n(m_1) y^n(j)}$, $j = 1, \ldots, m_2$
and finally decode $z^n(m_3)$ by applying a sequence of projectors
$\Pi_{x^n(m_1) y^n(m_2) z^n(k)}$, $k = 1, \ldots, m_3$.
Time sharing then gives us any
point in the interior of the rate region. However for several 
applications, including the application to the interference channel,
time sharing is not good enough and we seem to require a jointly typical
decoder that can achieve any point in the interior of the rate region
without time sharing and rate splitting. 

In Section~\ref{subsec:cmgmac}, we show that we can nevertheless
get an inner bound with a jointly typical sequential
decoder for a restricted MAC with three senders that we call the
Chong-Motani-Garg MAC. Our inner bound 
is larger than the inner bound proved earlier 
by Chong, Motani and Garg. However, it does not lead to a larger
inner bound for the interference channel as we will explain later in
more detail in Section~\ref{sec:ccqq}.

\subsection{The Chong-Motani-Garg three-sender MAC}
\label{subsec:cmgmac}
As a warm-up for the Chong-Motani-Garg
three-sender MAC, we first study a two-sender MAC where the receiver only
cares to decode correctly the message of the first sender. The senders
are forced to use the random encoding strategy of 
Figure~\ref{fig:cqencoding}. We shall call such a MAC as the 
{\em disinterested 2-MAC}. 
Classically, Costa and El Gamal~\cite{costaelgamal} considered
a disinterested 2-MAC in their work on the interference channel with
strong interference. They gave the following inner bound for the
disinterested 2-MAC.
\[
R_1 < I(X : B | Y), ~~~~~ R_1 + R_2 < I(X Y : B).
\]
Let $m_1$ and $m_2$ be the messages of senders~1 and 2 respectively.
The receiver is only interested in recovering $m_1$ correctly. 
Costa and El Gamal modelled this indifference towards 
$m_2$ by 
requiring the received word to be jointly typical with a pair 
$m_1, m_2$ where only $m_1$ has to be unique. In other words, the decoder
checks if the received word lies in the union, over all $m_2$, of sets of 
output words typical with $(m_1, m_2)$ for a given $m_1$.
If there is exactly one such message $m_1$, the receiver declares
it to be her guess for the message sent by sender~1, or else she declares 
error. This allowed them to `forget' the constraint on $R_2$ that
would normally be there in the inner bound Equation~\ref{eq:ccq}
for a standard 2-MAC.

In the quantum case, a naive analogue would be to loop over all messages
$m_1$ and for a fixed $m_1$, to try and project
the received state onto the span, over all $m_2$, of the conditionally 
typical subspaces of the pairs $(m_1, m_2)$. If the attempt were to 
succeed for a given message $m_1$, the receiver would declare it to be
her guess for the message sent by sender~1. However, it seems difficult 
to prove good upper bounds on the decoding error probability for this naive
quantum strategy. The main problem with this strategy is that in general 
the 
projector onto the span of a set of subspaces is not upper bounded well by
the sum of the projectors onto the individual subspaces, if those
projectors do not commute. This problem does not arise in the classsical 
case. Nevertheless, we can sidestep this problem by using a different
decoding strategy which will let us easily prove the following
inner bound for the disinterested 2-MAC. 
\begin{equation}
\begin{array}{c}
I(X : B)_\rho \leq R_1 < I(X : B | Y)_\rho, ~~~~~ 
R_1 + R_2 < I(XY : B)_\rho, \\
\mbox{OR} \\
R_1 < I(X : B)_\rho. \\
\end{array}
\label{eq:disinterestedccq}
\end{equation}
Above, $\rho^{XYB}$ is the joint state of the system $XYB$ defined
in Equation~\ref{eq:rhoXYB}. 
The above (quantum) inner bound is 
larger than Costa and El Gamal's inner bound. Our decoder is a 
{\em jointly typical}
decoder doing a sequence of orthogonal projections. 
In order to achieve a rate pair $(R_1, R_2)$ where $R_1 < I(X : B)_\rho$,
the decoder just adopts the sequential decoding procedure of
Figure~\ref{fig:cqdecoding}. 
In order to achieve a rate pair $(R_1, R_2)$ where 
$I(X : B)_\rho \leq R_1 < I(X : B | Y)_\rho$,
the decoder just adopts the sequential decoding procedure of
Figure~\ref{fig:ccqdecoding}. This works because now 
$R_2 < I(Y : B | X)$. Note that the above inner bound forms a 
non-convex set.

We now describe the Chong-Motani-Garg three-sender MAC, called CMG-MAC
for short.
Let $\chan: (z,y) \mapsto \rho_{zy}$ be a channel that takes two
classical inputs $z$ and $y$ and outputs a quantum state $\rho_{zy}$.
Let $X$, $Z$, $Y$ be classical random variables such that $Y$ is
independent of $XZ$. We define
\begin{equation}
\rho^{XZYB} := \sum_{(x,z,y) \in \cX \times \cZ \times \cY} 
p_X(x) p_{Z|X}(z|x) p_Y(y)
\ketbra{x,z,y} \otimes \rho_{zy}
\label{eq:rhoXZYB}
\end{equation}
to denote the joint state of the system $X Z Y B$.

In the CMG-MAC
there are three senders and one receiver. Suppose senders~1, 2, 3
want to transmit at rates $R_1$, $R_2$, $R_3$ using $n$ independent
uses of the channel $\chan$. Let their messages be denoted by
$m_1$, $m_2$, $m_3$, where $1 \leq m_i \leq 2^{nR_i}$. The senders are
forced to use the 
random encoding strategy of Figure~\ref{fig:cmgencoding}
 to choose a codebook $\codebook$. Note that the encoding strategies
of senders~1 and 2 are not independent in general.
\begin{figure}[!ht]
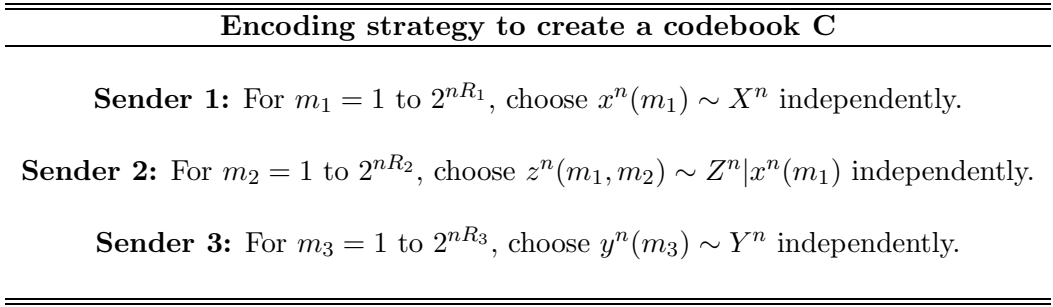

\begin{center}
\begin{tabular}{c}
\hline \hline 
{\bf Encoding strategy to create a codebook C} \\
\hline \\
{\bf Sender 1:}
For $m_1 = 1 \mbox{ to } 2^{nR_1}$, choose $x^n(m_1) \sim X^n$ 
independently.  \\ \\
{\bf Sender 2:}
For $m_2 = 1 \mbox{ to } 2^{nR_2}$, choose 
$z^n(m_1, m_2) \sim Z^n|x^n(m_1)$ independently. \\ \\
{\bf Sender 3:}
For $m_3 = 1 \mbox{ to } 2^{nR_3}$, choose $y^n(m_3) \sim Y^n$ 
independently. \\ \\
\hline \hline
\end{tabular}
\end{center}
\caption{Encoding strategy for CMG-MAC.}
\label{fig:cmgencoding}
\end{figure}

The receiver is only interested in recovering the message pair
$(m_1, m_2)$ correctly. As in our inner bound 
for the disinterested 2-MAC in Equation~\ref{eq:disinterestedccq},
we sidestep the problem of `forgetting' a constraint in the quantum 
setting by adopting one of two different decoding strategies depending upon
where our rate triple $(R_1, R_2, R_3)$ lies.
Our decoder is a {\em jointly typical}
decoder doing a sequence of orthogonal projections.
The inner bound that we end up proving via this two pronged decoding
strategy is larger than 
the inner bound proved
earlier in the classical setting by Chong, Motani and Garg. 
Note that our inner bound is not convex.
\begin{equation}
\label{eq:cmgmac}
\begin{array}{c}
R_3 < I(Y : B | Z)_\rho, ~~~~~
R_2 < I(Z : B | X Y)_\rho, ~~~~~ 
R_1 + R_2 < I(Z : B | Y)_\rho, \\
R_2 + R_3 < I(Z Y : B | X)_\rho, ~~~~~
R_1 + R_2 + R_3 < I(Z Y : B)_\rho, \\
\mbox{OR} \\
R_3 \geq I(Y : B | Z)_\rho, ~~~~~
R_2 < I(Z : B | X)_\rho, ~~~~~ 
R_1 + R_2 < I(Z : B)_\rho. 
\end{array}
\end{equation}
Recall that Chong, Motani and Garg's inner bound is
\[
R_2 < I(Z : B | X Y), ~~~ 
R_1 + R_2 < I(Z : B | Y), ~~~
R_2 + R_3 < I(Z Y : B | X), ~~~
R_1 + R_2 + R_3 < I(Z Y : B).
\]
Our inner bound contains their inner bound because
taking $R_3 \geq I(Y : B | Z)$ in their inner bound implies the
inequalities $R_2 < I(Z : B | X)$ and $R_1 + R_2 < I(Z :B)$. 

Define the following averaged density matrices:
\[
\begin{array}{c}
\rho_z := \sum_{y \in \cY} p_Y(y) \rho_{zy}, ~~~
\rho_{xy} := \sum_{z \in \cZ} p_{Z|X}(z|x) \rho_{zy}, ~~~
\rho_x := \sum_{z \in \cZ} p_{Z|X}(z|x) \rho_z =
          \sum_{y \in \cY} p_Y(y) \rho_{xy} \\
\rho_y := \sum_{z \in \cZ} p_Z(z) \rho_{zy}, ~~~
\rho   := \sum_{y \in \cY} p_Y(y) \rho_y.
\end{array}
\]
We can now define the $n$-fold tensor product matrices $\rho_{z^n y^n}$, 
$\rho_{x^n y^n}$, $\rho_{z^n}$, $\rho_{x^n}$, $\rho_{y^n}$ for sequences 
$x^n$, $z^n$, $y^n$ in the natural manner. 
Note that from the definition of $\rho^{XZYB}$ in
Equation~\ref{eq:rhoXZYB}, 
$
H(B|ZY) = H(B|ZXY), H(B|Z) = H(B|XZ).
$

Let $0 < \epsilon, \delta < 1/64$. Define 
\begin{eqnarray*}
p_{\mathrm{min}} 
& := & \min_{(x,z,y) \in \cX \times \cZ \times \cY, p_{XZ}(x,z) p_Y(y) > 0}
               p_{XZ}(x,z) p_Y(y),\\
q_{\mathrm{min}} 
& := & \min_{(x,z,y) \in \cX \times \cZ \times \cY, p_{XZ}(x,z) p_Y(y) > 0}
               q_{\mathrm{min}}(\rho_{z,y}).
\end{eqnarray*}
Let $n \geq 4 \delta^{-2} p_{\mathrm{min}}^{-1} q_{\mathrm{min}}^{-1}
     \log(|\cB| |\cX| |\cZ| |\cY| / \epsilon)$.
Define $c(\delta) := \delta \log |\cB||\cX||\cZ||\cY| - 
\delta \log \delta$. 
Suppose that $x^n z^n y^n \in T^{(XZY)^n}_\delta$. 
Let $\Pi_{z^n y^n}$, $\Pi_{x^n y^n}$, $\Pi_{y^n}$
denote the conditionally
typical projectors $\Pi^{\rho_{z^n y^n}}_\delta$, 
$\Pi^{\rho_{x^n y^n}}_{6\delta}$, 
$\Pi^{\rho_{y^n}}_{6\delta}$. 
We define a projector $\tPi_{x^n z^n y^n}$ as follows.
Treat the supports of the conditionally typical projectors 
$\Pi_{z^n y^n}$,
$\Pi_{x^n y^n}$ as the subspaces $A'$, $B'$ of Lemma~\ref{lem:chordal}, and
obtain $\tPi'_{x^n z^n y^n}$ as the projector $\Pi_B$
promised by Lemma~\ref{lem:chordal}. Now treat the supports of
$\tPi'_{x^n z^n y^n}$ and $\Pi_{y^n}$
as the subspaces $A'$, $B'$ of Lemma~\ref{lem:chordal}, and
obtain $\tPi_{x^n z^n y^n}$ as the projector $\Pi_B$
promised by Lemma~\ref{lem:chordal}.
If $x^n z^n y^n \not \in T^{(XZY)^n}_\delta$, define 
$\tPi_{x^n z^n y^n} := 0$. Then for all $x^n z^n y^n$,
\[
\tPi_{x^n z^n y^n} \leq 
(1-\sqrt{\epsilon})^{-1} \Pi_{y^n} \tPi'_{x^n z^n y^n} \Pi_{y^n} \leq
(1-\sqrt{\epsilon})^{-2} \Pi_{y^n}\Pi_{x^n y^n} \Pi_{z^n y^n} 
                         \Pi_{x^n y^n} \Pi_{y^n}.
\]
The decoding procedure shall use the projectors $\tPi_{x^n z^n y^n}$.

Depending upon whether
the desired rate triple $(R_1, R_2, R_3)$ falls into the first or
the second region in Equation~\ref{eq:cmgmac}, the receiver adopts
one of two sequential decoding strategies. 
As will become clear later, the ordering of messages
is not very important but we impose the lexicographic ordering
looping over $m_3$ first, then $m_2$ and then $m_1$, while describing the
decoding algorithms, for clarity of exposition. 

We now describe the receiver's decoding strategy
when the rate triple falls into the first region.
Given the codebook $\codebook$, the receiver uses the sequential
decoding strategy of Figure~\ref{fig:cmgdecoding1}.
\begin{figure}[!ht]
\begin{center}
\begin{tabular}{c}
\hline \hline 
{\bf Decoding strategy for codebook C for first region} \\
\hline \\
\begin{minipage}[t]{15cm}
\begin{itemize}
\item Received some state $\tau^{B^n}$ at the channel output;

\item 
\begin{minipage}[t]{13.5cm}
\begin{tabbing}
For \= $m_1 = 1$ to  $2^{nR_1}$, \\
\> For \= $m_2 = 1$ to $2^{nR_2}$, \\
\> \> For \= $m_3 = 1$ to $2^{nR_3}$, \\
\> \> \> Repeat the following till first success; \\
\> \> \> \begin{minipage}[t]{11.9cm}
         \begin{itemize}
         \item Let $\tau_{m_1 m_2 m_3}$ be the state at the beginning of 
               loop number $(m_1, m_2, m_3)$;
     
      \item Project $\tau_{m_1 m_2 m_3}$ onto 
            $\tPi_{x^n(m_1), z^n(m_1,m_2), y^n(m_3)}$.
            If fail proceed, else
            declare success and announce $(m_1, m_2, m_3)$ to be the 
            guess for 
            the messages sent by senders~1, 2 and 3.
      \end{itemize}
      \end{minipage}
\end{tabbing}
\end{minipage}

\item If all candidate message triples are exhausted without success, 
      abort.
\end{itemize}
\end{minipage} \\ \\
\hline \hline
\end{tabular}
\end{center}
\caption{Decoding strategy for CMG-MAC for first region.}
\label{fig:cmgdecoding1}
\end{figure}

In order to analyse the success probability of decoding a message
correctly, we first imagine that we are given a slightly 
different channel $\chan'$ working on alphabets $\cX^n$, $\cZ^n$, 
$\cY^n$ with
output Hilbert space $\cB^{\otimes n}$.  It is defined as
$\chan': (x^n, z^n, y^n) \mapsto \rho'_{x^n z^n y^n}$ where
the states $\{\rho'^{B^n}_{x^n z^n y^n}\}_{(x^n, z^n, y^n) \in 
\cX^n \times \cZ^n \times \cY^n}$
are as provided by Lemma~\ref{lem:asymsmooth} given state 
$\rho^{XZYB}$. 
Let $\rho'^{X^n Z^n Y^n B^n}$ be as defined by Lemma~\ref{lem:asymsmooth}. 
The encoding and decoding strategies used continue to be exactly as in 
Figures~\ref{fig:cmgencoding} and \ref{fig:cmgdecoding1}, even though
the channel we are now working with is a single copy of $\chan'$.

We shall use the notation
$(x^n,z^n,y^n)(m_1,m_2,m_3)$ to denote
$(x^n(m_1), z^n(m_1,m_2), y^n(m_3))$. 
We shall use the symbol $\E[\cdot]$ to denote expectation 
over the choice of the random codebook $\codebook$.
We shall also use notation like 
$\E_{x^n z^n y^n} [\rho'_{x^n z^n y^n}] :=
 \sum_{x^n, z^n, y^n} p_{(XZ)^n}(x^n z^n) p_{Y^n}(y^n) 
     \rho'_{x^n z^n y^n}$, and
$\E_{z^n|x^n}[\rho'_{x^n z^n y^n}] := 
 \sum_{z^n} p_{Z^n|X^n}(z^n|x^n) \rho'_{x^n z^n y^n}.
$
Arguing as in Section~\ref{subsec:ccq}, we see that
the expected error 
probability $\E[\err(m_1,m_2,m_3)]$ of decoding $(m_1,m_2,m_3)$ is at 
most
\begin{eqnarray*}
\lefteqn{\E[\err(m_1,m_2,m_3)]} \\
& \leq & 2 \cdot
         \E\left[
         \left(
         \begin{array}{c}
         \displaystyle{
         \sum_{\stackrel{(i,j,k):}{(i,j,k) \neq (m_1,m_2,m_3)}}
         }
         \Tr [\tPi_{(x^n,z^n,y^n)(i,j,k)} 
              \rho'_{(x^n,z^n,y^n)(m_1,m_2,m_3)}] + 1 \\
         {} - \Tr [\tPi_{(x^n,z^n,y^n)(m_1,m_2,m_3)} 
                   \rho'_{(x^n,z^n,y^n)(m_1,m_2,m_3)}]
         \end{array}
         \right)^{1/2}
         \right] \\
& \leq & 2 \cdot
         \left(
         \begin{array}{c}
         \displaystyle{\sum_{(i,j): i \neq m_1}}
         \Tr [\E[\tPi_{(x^n,z^n,y^n)(i,j,m_3)} 
                 \rho'_{(x^n,z^n,y^n)(m_1,m_2,m_3)}]] + {} \\
         \displaystyle{\sum_{j: j \neq m_2}}
         \Tr [\E[\tPi_{(x^n,z^n,y^n)(m_1,j,m_3)} 
              \rho'_{(x^n,z^n,y^n)(m_1,m_2,m_3)}]] + {} \\
         \displaystyle{\sum_{k: k \neq m_3}}
         \Tr [\E[\tPi_{(x^n,z^n,y^n)(m_1,m_2,k)} 
                 \rho'_{(x^n,z^n,y^n)(m_1,m_2,m_3)}]] + {} \\
         \displaystyle{\sum_{\stackrel{(j,k):}{j \neq m_2, k \neq m_3}}}
         \Tr [\E[\tPi_{(x^n,z^n,y^n)(m_1,j,k)} 
                 \rho'_{(x^n,z^n,y^n)(m_1,m_2,m_3)}]] + {} \\
         \displaystyle{\sum_{\stackrel{(i,j,k):} {i \neq m_1, k \neq m_3}}}
         \Tr [\E[\tPi_{(x^n,z^n,y^n)(i,j,k)} 
                 \rho'_{(x^n,z^n,y^n)(m_1,m_2,m_3)}]] + {} \\
         1 - 
         \E[\Tr [\tPi_{(x^n,z^n,y^n)(m_1,m_2,m_3)} 
                 \rho'_{(x^n,z^n,y^n)(m_1,m_2,m_3)}]]
         \end{array}
         \right)^{1/2} \\
& \leq & 2 \cdot
         \left(
         \begin{array}{l l}
         2^{n(R_1 + R_2)} \cdot
         \E_{x^n z^n y^n}[\Tr [\tPi_{x^n z^n y^n} \rho'_{y^n}]] + 
         2^{n R_2} \cdot
         \E_{x^n z^n y^n}[\Tr [\tPi_{x^n z^n y^n} \rho'_{x^n y^n}]] + {} \\
         2^{n R_3} \cdot
         \E_{x^n z^n y^n}[\Tr [\tPi_{x^n z^n y^n} \rho'_{x^n z^n}]] + 
         2^{n (R_2 + R_3)} \cdot
         \E_{x^n z^n y^n}[\Tr [\tPi_{x^n z^n y^n} \rho'_{x^n}]] + {} \\
         2^{n (R_1 + R_2 + R_3)} \cdot
         \E_{x^n z^n y^n}[\Tr [\tPi_{x^n z^n y^n} \rho']] + 
         1 - 
         \E_{x^n z^n y^n}[\Tr [\tPi_{x^n z^n y^n} \rho'_{x^n z^n y^n}]]
         \end{array}
         \right)^{1/2} \\
\end{eqnarray*}
The second inequality follows from concavity of square root. The
third inequality follows from the independent and identical choice
of codewords for a pair of different messages and from the encoding
strategy used.

We now bound the terms on the right
hand side above. 
Arguing as in Section~\ref{subsec:ccq}, we see that
for any $y^n \in T^{Y^n}_\delta$,
\begin{eqnarray*}
\lefteqn{\E_{x^n z^n}[\Tr [\tPi_{x^n z^n y^n} \rho'_{y^n}]]} \\
& \leq & (1-\sqrt{\epsilon})^{-2}
         \E_{x^n z^n}[
         \Tr [\Pi_{y^n} \Pi_{x^n y^n} \Pi_{z^n y^n} \Pi_{x^n y^n} 
              \Pi_{y^n} \rho'_{y^n}]
         ] \\
& \leq & 2^{n(H(B|ZY)_\rho + c(\delta))} (1-\sqrt{\epsilon})^{-2}
         \E_{x^n z^n}[
         \Tr [\Pi_{y^n} \Pi_{x^n y^n} \rho_{z^n y^n} \Pi_{x^n y^n} 
              \Pi_{y^n} \rho'_{y^n}]
         ] \\
&   =  & 2^{n(H(B|ZY)_\rho + c(\delta))} (1-\sqrt{\epsilon})^{-2} 
         \E_{x^n}[
         \Tr [\Pi_{y^n} \Pi_{x^n y^n} (\E_{z^n|x^n}[\rho_{z^n y^n}])
              \Pi_{x^n y^n} \Pi_{y^n} \rho'_{y^n}]
         ] \\
&   =  & 2^{n(H(B|ZY)_\rho + c(\delta))} (1-\sqrt{\epsilon})^{-2}
         \E_{x^n}[
         \Tr [\Pi_{y^n} \Pi_{x^n y^n} \rho_{x^n y^n} \Pi_{x^n y^n}
              \Pi_{y^n} \rho'_{y^n}]
         ] \\
& \leq & 2^{n(H(B|ZY)_\rho + c(\delta))} (1-\sqrt{\epsilon})^{-2}
         \E_{x^n}[\Tr [\Pi_{y^n} \rho_{x^n y^n} \Pi_{y^n} \rho'_{y^n}]] \\
&   =  & 2^{n(H(B|ZY)_\rho + c(\delta))} (1-\sqrt{\epsilon})^{-2}
         \Tr [\Pi_{y^n} \rho_{y^n} \Pi_{y^n} \rho'_{y^n}] \\
& \leq & 2^{-n(H(B|Y)_\rho - c(6\delta))} \cdot 2^{n(H(B|ZY)+c(\delta))}
         (1-\sqrt{\epsilon})^{-2}
         \Tr [\Pi_{y^n} \rho'_{y^n}] \\
& \leq & 2^{-n(I(Z:B|Y)_\rho - 2c(6\delta))} (1-\sqrt{\epsilon})^{-2}
\;\leq\; 2 \cdot 2^{-n(I(Z:B|Y)_\rho - 2c(6\delta))}.
\end{eqnarray*}
The above upper bound holds trivially if $y^n \not \in T^{Y^n}_\delta$.

Similarly, for any $x^n y^n \in T^{(XY)^n}_\delta$,
\begin{eqnarray*}
\lefteqn{\E_{z^n|x^n}[\Tr [\tPi_{x^n z^n y^n} \rho'_{x^n y^n}]]} \\
& \leq & (1-\sqrt{\epsilon})^{-2}
         \E_{z^n|x^n}[
         \Tr [\Pi_{y^n} \Pi_{x^n y^n} \Pi_{z^n y^n} \Pi_{x^n y^n} 
              \Pi_{y^n} \rho'_{x^n y^n}]
         ] \\
& \leq & 2^{n(H(B|ZY)_\rho + c(\delta))} (1-\sqrt{\epsilon})^{-2}
         \E_{z^n|x^n}[
         \Tr [\Pi_{y^n} \Pi_{x^n y^n} \rho_{z^n y^n} \Pi_{x^n y^n} 
              \Pi_{y^n} \rho'_{x^n y^n}]
         ] \\
&   =  & 2^{n(H(B|ZY)_\rho + c(\delta))} (1-\sqrt{\epsilon})^{-2}
         \Tr [\Pi_{y^n} \Pi_{x^n y^n} \rho_{x^n y^n} \Pi_{x^n y^n} 
              \rho'_{x^n y^n}] \\
& \leq & 2^{-n(H(B|XY)_\rho - c(6\delta))} \cdot 2^{n(H(B|ZY)+c(\delta))}
         (1-\sqrt{\epsilon})^{-2}
         \Tr [\Pi_{y^n} \Pi_{x^n y^n} \Pi_{y^n} \rho'_{x^n y^n}] \\
& \leq & 2^{-n(H(B|XY)_\rho - c(6\delta))} \cdot 2^{n(H(B|ZXY)+c(\delta))}
         (1-\sqrt{\epsilon})^{-2} \\
& \leq & 2 \cdot 2^{-n(I(Z:B|XY)_\rho - 2c(6\delta))}.
\end{eqnarray*}
The above upper bound holds trivially if 
$x^n y^n \not \in T^{(XY)^n}_\delta$.

Next, we see that for any 
$x^n z^n y^n \in \cX^n \times \cZ^n \times \cY^n$,
\begin{eqnarray*}
\Tr [\tPi_{x^n z^n y^n} \rho'_{x^n z^n}]
& \leq & \ellinfty{\rho'_{x^n z^n}} \cdot \Tr [\tPi_{x^n z^n y^n}] \\
& \leq & 4 \cdot 2^{-n(H(B|XZ)_\rho - c(6\delta))} \cdot
         2^{n(H(B|ZY)_\rho + c(\delta))} \\
&   =  & 4 \cdot 2^{-n(H(B|Z)_\rho - c(6\delta))} \cdot
         2^{n(H(B|ZY)_\rho + c(\delta))} 
\;\leq\; 4 \cdot 2^{-n(I(Y:B|Z)_\rho - 2c(6\delta))},
\end{eqnarray*}
where we use Lemma~\ref{lem:asymsmooth}. 

Similarly, for any
$x^n z^n y^n \in \cX^n \times \cZ^n \times \cY^n$,
\begin{eqnarray*}
\Tr [\tPi_{x^n z^n y^n} \rho'_{x^n}]
& \leq & \ellinfty{\rho'_{x^n}} \cdot \Tr [\tPi_{x^n z^n y^n}] \\
& \leq & 4 \cdot 2^{-n(H(B|X)_\rho - c(6\delta))} \cdot
         2^{n(H(B|ZY)_\rho + c(\delta))} \\
&   =  & 4 \cdot 2^{-n(H(B|X)_\rho - c(6\delta))} \cdot
         2^{n(H(B|ZXY)_\rho + c(\delta))} 
\;  = \; 4 \cdot 2^{-n(I(YZ:B|X)_\rho - 2c(6\delta))}.
\end{eqnarray*}

Again, for any
$x^n z^n y^n \in \cX^n \times \cZ^n \times \cY^n$,
\begin{eqnarray*}
\Tr [\tPi_{x^n z^n y^n} \rho']
& \leq & \ellinfty{\rho'} \cdot \Tr [\tPi_{x^n z^n y^n}] \\
& \leq & 4 \cdot 2^{-n(H(B)_\rho - c(2\delta))} \cdot
         2^{n(H(B|ZY)_\rho + c(\delta))} 
\;  = \; 4 \cdot 2^{-n(I(YZ:B)_\rho - 2c(2\delta))}.
\end{eqnarray*}

Finally, for $x^n z^n y^n \in T^{(XZY)^n}_\delta$,
\begin{eqnarray*}
\lefteqn{\Tr [\tPi'_{x^n z^n y^n} \rho'_{x^n z^n y^n}]}\\
& \geq & \Tr [\tPi'_{x^n z^n y^n} \Pi_{z^n y^n} 
              \rho_{z^n y^n} \Pi_{z^n y^n}] - 
         \ellone{\rho_{z^n y^n} - 
                 \Pi_{z^n y^n} \rho_{z^n y^n} \Pi_{z^n y^n}} - 
         \ellone{\rho_{z^n y^n} - \rho'_{x^n z^n y^n}} \\
& \geq & \Tr [\tPi'_{x^n z^n y^n} \Pi_{z^n y^n} 
              \rho_{z^n y^n} \Pi_{z^n y^n}] - 
         2\epsilon - 11\sqrt{\epsilon}
\;\geq\; \Tr [\tPi'_{x^n z^n y^n} \Pi_{z^n y^n} 
              \rho_{z^n y^n} \Pi_{z^n y^n}] - 13 \sqrt{\epsilon},
\end{eqnarray*}
where we used Lemma~\ref{lem:asymsmooth}, 
Fact~\ref{fact:condtypicalspace} and the fact that $\Pi_{z^n y^n}$
commutes with $\rho_{z^n y^n}$. Since
\[
\Tr [\Pi_{x^n y^n} (\Pi_{z^n y^n} \rho_{z^n y^n} \Pi_{z^n y^n})] \geq
\Tr [\Pi_{x^n y^n} \rho_{z^n y^n}] - 
\ellone{\rho_{z^n y^n} - \Pi_{z^n y^n} \rho_{z^n y^n} \Pi_{z^n y^n}} \geq
1 - \epsilon - 2\epsilon = 1 - 3\epsilon,
\]
where we used Facts~\ref{fact:winter}, \ref{fact:condtypicalspace} and 
the fact that $\Pi_{z^n y^n}$ commutes with $\rho_{z^n y^n}$, we see that
\[
\Tr [\tPi'_{x^n z^n y^n} (\Pi_{z^n y^n} \rho_{z^n y^n}
     \Pi_{z^n y^n})] \geq 1 - 4\sqrt{\epsilon},
\]
by Lemma~\ref{lem:chordal}. Hence,
$
\Tr [\tPi'_{x^n z^n y^n} \rho'_{x^n z^n y^n}]
\geq 1 - 17\sqrt{\epsilon}.
$
Now 
\begin{eqnarray*}
\lefteqn{
\Tr [\Pi_{y^n}(\tPi'_{x^n z^n y^n}\rho'_{x^n z^n y^n}\tPi'_{x^n z^n y^n})]
}\\
& \geq & \Tr [\Pi_{y^n} \rho_{z^n y^n}] - 
         \ellone{\rho'_{x^n z^n y^n} - 
                 \tPi'_{x^n z^n y^n} \rho'_{x^n z^n y^n} 
                 \tPi'_{x^n z^n y^n}} - 
         \ellone{\rho_{z^n y^n} - \rho'_{x^n z^n y^n}} \\
& \geq & 1 - \epsilon - 10\epsilon^{1/4} - 11\sqrt{\epsilon} 
\;\geq\; 1 - 22\epsilon^{1/4},
\end{eqnarray*}
where we used Lemma~\ref{lem:asymsmooth} and 
Facts~\ref{fact:condtypicalspace}, \ref{fact:gentle}. By 
Lemma~\ref{lem:chordal},
\[
\Tr [\tPi_{x^n y^n z^n}
     (\tPi'_{x^n z^n y^n}\rho'_{x^n z^n y^n}\tPi'_{x^n z^n y^n})] \geq
1 - 10\epsilon^{1/8},
\]
which implies
\begin{eqnarray*}
\lefteqn{\Tr [\tPi_{x^n y^n z^n} \rho'_{x^n z^n y^n}]} \\
& \geq &
\Tr [\tPi_{x^n y^n z^n}
     (\tPi'_{x^n z^n y^n}\rho'_{x^n z^n y^n}\tPi'_{x^n z^n y^n})] -
\ellone{\rho'_{x^n y^n z^n} -
        \tPi'_{x^n z^n y^n}\rho'_{x^n z^n y^n}\tPi'_{x^n z^n y^n}} \\
& \geq &
1 - 10\epsilon^{1/8} - 10\epsilon^{1/4} 
\;\geq\;
1 - 20\epsilon^{1/8}.
\end{eqnarray*}
We proceed to get
\[
\E_{x^n z^n y^n}[\Tr [\tPi_{x^n z^n y^n} \rho'_{x^n z^n y^n}]] \geq
(1 - \epsilon) (1 - 20\epsilon^{1/4}) \geq
1 - 21\epsilon^{1/8},
\]
using Fact~\ref{fact:typicalset}.

Putting all this together, we get
\begin{eqnarray*}
\lefteqn{\E[\err(m_1,m_2,m_3)]} \\
& \leq & 2 \cdot
         \left(
         \begin{array}{l l}
         2^{n(R_1 + R_2)} \cdot
         2 \cdot 2^{-n(I(Z:B|Y)_\rho - 2c(6\delta))} +
         2^{n R_2} \cdot
         2 \cdot 2^{-n(I(Z:B|XY)_\rho - 2c(6\delta))} + {} \\
         2^{n R_3} \cdot
         4 \cdot 2^{-n(I(Y:B|Z)_\rho - 2c(6\delta))} +
         2^{n (R_2 + R_3)} \cdot
         4 \cdot 2^{-n(I(YZ:B|X)_\rho - 2c(6\delta))} + {} \\
         2^{n (R_1 + R_2 + R_3)} \cdot
         4 \cdot 2^{-n(I(YZ:B)_\rho - 2c(6\delta))} +
         21\epsilon^{1/8}
         \end{array}
         \right)^{1/2}.
\end{eqnarray*}
Choosing rate triples satisfying the inequalities
\begin{equation}
\begin{array}{c}
R_1 + R_2 \leq I(Z:B|Y)_\rho - 4 c(6\delta),
R_2 \leq I(Z:B|XY)_\rho - 4 c(6\delta), 
R_3 \leq I(Y:B|Z)_\rho - 4 c(6\delta), \\
R_2 + R_3 \leq I(YZ:B|X)_\rho - 4 c(6\delta), ~~~
R_1 + R_2 + R_3 \leq I(YZ:B)_\rho - 4 c(6\delta), 
\end{array}
\label{eq:cmgmacprime1}
\end{equation}
gives us
$\E[\err(m_1,m_2,m_3)] \leq 10\epsilon^{1/16} + 8 \cdot 2^{-nc(6\delta)}$,
for which if we take $n = C \log (1/\epsilon) \delta^{-2}$ 
for a constant
$C$ depending only on the state $\rho^{XZYB}$, we get
\[
\E[\err(m_1,m_2)] \leq \E[\err(m_1,m_2,m_3)] \leq 10\epsilon^{1/16} + 
  8 \cdot 2^{-C\delta^{-2} c(6\delta) \log(1/\epsilon)}.
\]

We now describe the receiver's decoding strategy when the rate triple
falls into the second region. 
Define $\tPi_{x^n z^n} := \Pi_{z^n} := \Pi^{\rho_{z^n}}_{\delta}$ if 
$x^n z^n \in T^{(XZ)^n}_\delta$, and $\tPi_{x^n z^n} := 0$ otherwise.
The receiver uses the sequential
decoding strategy of Figure~\ref{fig:cmgdecoding2}. 
\begin{figure}[!ht]
\begin{center}
\begin{tabular}{c}
\hline \hline 
{\bf Decoding strategy for codebook C for second region} \\
\hline \\
\begin{minipage}[t]{15cm}
\begin{itemize}
\item Received some state $\tau^{B^n}$ at the channel output;

\item 
\begin{minipage}[t]{13.5cm}
\begin{tabbing}
For \= $m_1 = 1$ to  $2^{nR_1}$, \\
\> For \= $m_2 = 1$ to $2^{nR_2}$, \\
\> \>  Repeat the following till first success; \\
\> \>  \begin{minipage}[t]{11.9cm}
         \begin{itemize}
         \item Let $\tau_{m_1 m_2}$ be the state at the beginning of 
               loop number $(m_1, m_2)$;
     
      \item Project $\tau_{m_1 m_2}$ onto 
            $\tPi_{(x^n,z^n)(m_1,m_2)}$.
            If fail proceed, else
            declare success and announce $(m_1, m_2)$ to be the 
            guess for 
            the messages sent by senders~1 and 2.
      \end{itemize}
      \end{minipage}
\end{tabbing}
\end{minipage}

\item If all candidate message triples are exhausted without success, 
      abort.
\end{itemize}
\end{minipage} \\ \\
\hline \hline
\end{tabular}
\end{center}
\caption{Decoding strategy for CMG-MAC for second region.}
\label{fig:cmgdecoding2}
\end{figure}

Again, we analyse the decoding error probability for the modified
channel $\chan'$ first.
Arguing as in Section~\ref{subsec:ccq}, we see that the expected error 
probability $\E[\err(m_1,m_2)]$ of decoding $(m_1,m_2)$ is now at 
most
\begin{eqnarray*}
\lefteqn{\E[\err(m_1,m_2)]} \\
& \leq & 2 \cdot
         \E\left[
         \left(
         \begin{array}{c}
         \displaystyle{\sum_{(i,j): (i,j) \neq (m_1,m_2)}}
         \Tr [\tPi_{(x^n,z^n)(i,j)} 
                 \rho'_{(x^n,z^n,y^n)(m_1,m_2,m_3)}] + {} \\
         1 - 
         \Tr [\tPi_{(x^n,z^n)(m_1,m_2)} 
                 \rho'_{(x^n,z^n,y^n)(m_1,m_2,m_3)}]
         \end{array}
         \right)^{1/2} 
         \right] \\
& \leq & 2 \cdot
         \left(
         \begin{array}{c}
         \displaystyle{\sum_{(i,j): i \neq m_1}}
         \Tr [\E[\tPi_{(x^n,z^n)(i,j)} 
                 \rho'_{(x^n,z^n,y^n)(m_1,m_2,m_3)}]] + {} \\
         \displaystyle{\sum_{j: j \neq m_2}}
         \Tr [\E[\tPi_{(x^n,z^n)(m_1,j)} 
              \rho'_{(x^n,z^n,y^n)(m_1,m_2,m_3)}]] + {} \\
         1 - 
         \E[\Tr [\tPi_{(x^n,z^n)(m_1,m_2)} 
                 \rho'_{(x^n,z^n,y^n)(m_1,m_2,m_3)}]]
         \end{array}
         \right)^{1/2} \\
& \leq & 2 \cdot
         \left(
         \begin{array}{l l}
         2^{n(R_1 + R_2)} \cdot
         \E_{x^n z^n}[\Tr [\tPi_{x^n z^n} \rho']] + 
         2^{n R_2} \cdot
         \E_{x^n z^n}[\Tr [\tPi_{x^n z^n} \rho'_{x^n}]] + {} \\
         1 - 
         \E_{x^n z^n}[\Tr [\tPi_{x^n z^n} \rho'_{x^n z^n}]]
         \end{array}
         \right)^{1/2}.
\end{eqnarray*}

For any $(x^n, z^n) \in \cX^n \times \cZ^n$,
\[
\Tr [\tPi_{x^n z^n} \rho']
  \leq   \ellinfty{\rho'} \cdot \Tr [\tPi_{x^n z^n}]
  \leq   2^{-n(H(B)_\rho - c(2\delta)} \cdot 
         2^{n(H(B|Z)_\rho + c(\delta))} 
  \leq   4 \cdot 2^{-n(I(Z:B)_\rho - 2c(2\delta))}.
\]

Also, for any
$(x^n, z^n) \in \cX^n \times \cZ^n$,
\begin{eqnarray*}
\Tr [\tPi_{x^n z^n} \rho'_{x^n}]
& \leq & \ellinfty{\rho'_{x^n}} \cdot \Tr [\tPi_{x^n z^n}] 
\;\leq\; 4 \cdot 2^{-n(H(B|X)_\rho - c(6\delta))} \cdot
         2^{n(H(B|Z)_\rho + c(\delta))} \\
&   =  & 4 \cdot 2^{-n(H(B|X)_\rho - c(6\delta))} \cdot
         2^{n(H(B|ZX)_\rho + c(\delta))}
\;  = \; 4 \cdot 2^{-n(I(Z:B|X)_\rho - 2c(6\delta))}. 
\end{eqnarray*}

Finally, for $x^n z^n \in T^{(XZ)^n}_\delta$,
\begin{eqnarray*}
\lefteqn{\Tr [\tPi_{x^n z^n} \rho'_{x^n z^n}]} \\
& \geq &
\Tr [\Pi_{z^n} \rho_{z^n}] - \ellone{\rho_{z^n} - \rho'_{x^n z^n}} 
\;\geq\;
\Tr [\Pi_{z^n} \rho_{z^n}] - 
\E_{y^n}[\ellone{\rho_{z^n y^n} - \rho'_{x^n z^n y^n}}] \\
& \geq &
1 - \epsilon - 11\sqrt{\epsilon}
\;\geq\;
1 - 12\sqrt{\epsilon},
\end{eqnarray*}
where we used Lemma~\ref{lem:asymsmooth}. Using Fact~\ref{fact:typicalset},
we get
\[
\E_{x^n z^n}[\Tr [\tPi_{x^n z^n} \rho'_{x^n z^n}]] \geq
(1-\epsilon) (1 - 12\sqrt{\epsilon}) \geq
1 - 13\sqrt{\epsilon}.
\]

Putting it together, we get
\begin{eqnarray*}
\lefteqn{\E[\err(m_1,m_2)]} \\
& \leq & 2 \cdot
         \left(
         2^{n(R_1 + R_2)} \cdot
         4 \cdot 2^{-n(I(Z:B)_\rho - 2c(6\delta))} +
         2^{n R_2} \cdot
         4 \cdot 2^{-n(I(Z:B|X)_\rho - 2c(6\delta))} +
         13\sqrt{\epsilon}
         \right)^{1/2}.
\end{eqnarray*}
Choosing rate triples satisfying the inequalities
\begin{equation}
R_1 + R_2 \leq I(Z:B)_\rho - 4 c(6\delta), ~~~
R_2 \leq I(Z:B|X)_\rho - 4 c(6\delta), 
\label{eq:cmgmacprime2}
\end{equation}
gives us
$\E[\err(m_1,m_2)] \leq 8\epsilon^{1/4} + 8 \cdot 2^{-nc(6\delta)}$,
for which if we take $n = C \log (1/\epsilon) \delta^{-2}$ 
for a constant
$C$ depending only on the state $\rho^{XZYB}$, we get
\[
\E[\err(m_1,m_2)] \leq 8\epsilon^{1/4} + 
  8 \cdot 2^{-C\delta^{-2} c(6\delta) \log(1/\epsilon)}.
\]

Note that in the second decoding strategy, the rate $R_3$ plays no
role. However, if $R_3 < I(Y:B|Z)_\rho$, the rate triple falls into
region~1. Hence, we impose the condition $R_3 \geq I(Y:B|Z)_\rho$ while
describing region~2.

Since we have shown that $\E[\err(m_1,m_2)]$ is small for all 
message pairs $(m_1,m_2)$ for both region~1 and 2, the expected average 
error probability 
$
\E[\avgerr_{\chan'}] := 
\E[2^{-n(R_1+R_2)} \sum_{m_1,m_2} \err_{\chan'}(m_1,m_2)]
$ 
for channel $\chan'$ is also small.
We can now apply an argument similar to that of the proof of
Theorem~\ref{thm:cq} to show that the expected average error probability
is also small for $n$ copies of the original
channel $\chan$ for the same encoding and decoding procedures as
described in Figures~\ref{fig:cmgencoding}, \ref{fig:cmgdecoding1},
\ref{fig:cmgdecoding2}. Doing so allows us to show that
\[
\E[\avgerr_{\chan}] \leq \E[\avgerr_{\chan'}] + 13\sqrt{\epsilon}
\leq 23\epsilon^{1/16} + 
8 \cdot 2^{-C\delta^{-2} c(6\delta) \log(1/\epsilon)},
\]
where $C$ is the same constant as above, depending only on the channel
$\chan$. This quantity can be made
arbitrarily small by choosing $\epsilon$ and $\delta$ appropriately.
We can also simultaneously make $c(6\delta)$ arbitrarily small so that
the rate pair $(R_1, R_2)$ approaches the boundary of the region
described in Equation~\ref{eq:cmgmac}.
We have thus shown the following theorem.
\begin{theorem}
For the CMG-MAC with classical input and
quantum output,
the random encoding and sequential decoding strategies described in
Figures~\ref{fig:cmgencoding}, and Figures~\ref{fig:cmgdecoding1} and
\ref{fig:cmgdecoding2} can 
achieve any rate in the region described by 
Equation~\ref{eq:cmgmac} with arbitrary small expected
average probability of error in the asymptotic limit
of many independent uses of the channel.
\label{thm:cmgmac}
\end{theorem}

\paragraph{A pretty good measurement decoder:}
Region~1 of our inner bound for the CMG-MAC can also be achieved by
using the pretty good measurement with positive operators 
\[
\Pi_{y^n(m_3)} \Pi_{x^n(m_1) y^n(m_3)} \Pi_{z^n(m_1, m_2) y^n(m_3)} 
\Pi_{x^n(m_1) y^n(m_3)} \Pi_{y^n(m_3)}
\]
for typical $x^n(m_1) z^n(m_1, m_2) y^n(m_3)$ and operator $0$ for
atypical $x^n(m_1) z^n(m_1, m_2) y^n(m_3)$.
Region~2 of our inner bound for the CMG-MAC can also be achieved by
using the pretty good measurement with positive operators 
$\Pi_{z^n(m_1, m_2)}$ for typical $x^n(m_1) z^n(m_1, m_2)$ and 
operator $0$ for atypical $x^n(m_1) z^n(m_1, m_2)$.
We analyse the decoding error as above, imagining
that the channel we are working with is a single copy of $\chan'$ and
then concluding that the decoding error for $n$ copies of the original
channel $\chan$ is small. 

\section{Sequential decoding for a two user interference channel}
\label{sec:ccqq}
The best known inner bound for a general two-sender two-receiver
interference channel in the classical setting is the Han-Kobayashi
rate region. To achieve a rate pair in this region requires one
to construct a jointly typical decoder for a three sender MAC. 
Fawzi et al.~\cite{mcgill:qic} showed that the same
technique would extend to a two-sender two-receiver interference
channel with classical inputs and quantum outputs (ccqq-IC), if we 
had a jointly typical decoder for a three-sender
MAC with classical inputs and quantum output. They conjectured the 
existence of such a jointly typical decoder and proved it for a special
case where some averaged output states commute. They also proved some
other variants of the conjecture. However they left open
the general case of the conjecture. In this work, we make some progress
on the conjecture by proving it with reduced commutativity requirements,
but we have not yet been able to prove the general case. Nevertheless,
we have constructed a jointly typical decoder for the CMG-MAC
in Section~\ref{subsec:cmgmac}. Using this decoder,
we prove the Chong-Motani-Garg inner bound
for the ccqq-IC in this section. The Chong-Motani-Garg inner bound
is known to be equivalent to the Han-Kobayashi inner bound under the
addition of a time-sharing public random variable and going through
all possible joint probability distributions of the time sharing variable
and the input random variables~\cite{cmgelgamal}.

Let $\chan: (x, y) \mapsto \rho_{x,y}$ be a 
ccqq-IC i.e. channel $\chan$ takes two classical inputs $x$, $y$ and
gives a quantum output $\rho_{x,y}$ in the Hilbert space 
$\cB_1 \otimes \cB_2$ of the two receivers. Let
$X$, $Y$ be the input random variables of the channel. Like
Chong, Motani and Garg, we introduce three additional random variables: 
a `time-sharing'
random variable $Q$ which is public to all senders and receivers, and
random variables $U$, $V$ which represent the `common' parts of 
messages of senders 1 and 2. In other words, we consider a joint
probability distribution on 
$\cQ \times \cU \times \cX \times \cV \times \cY$ that factors as
$p_{QUXVY}(q, u, x, v, y) = p_Q(q) p_{UX|Q}(u, x | q) p{VY|Q}(v, y|q)$. 
Let the joint state of the system $QUXVY B_1 B_2$ be
\begin{eqnarray*}
\rho^{QUXVYB_1B_2} :=  
\sum_{(q,u,x,v,y) \in \cQ \times \cU \times \cX \times \cV \times \cY} 
&
p_Q(q) p_{UX|Q}(u,x|q) p_{VY|Q}(v,y|q) \\
& 
\ketbra{q,u,x,v,y} \otimes \rho^{B_1 B_2}_{x,y}.
\end{eqnarray*}

Suppose senders 1 and 2 want to transmit to receivers 1 and 2 at rates
$R_1$ and $R_2$ using $n$ independent uses of the channel $\chan$. 
Sender 1 splits her rate $R_1$ into $R_{1c}$ for the `common' part
and $R_{1p}$ for the `personal' part. In other words, her messages
come from the set $2^{nR_{1c}} \times 2^{nR_{1p}}$ where 
$R_{1c} + R_{1p} = R_1$. Let $m_{1c}$, $m_{1p}$ denote the two parts
of the message of sender 1. A similar `rate-splitting' strategy is 
adopted by sender 2. Receiver~1 wants to decode the complete message
of Sender~1, and for some rate points, also the common message of
Sender~2. Receiver~2 also has a similar aim.
The senders use the random encoding strategy of
Figure~\ref{fig:ccqqencoding} in order to transmit information over the 
channel $\chan$.
\begin{figure}[!ht]
\begin{center}
\begin{tabular}{l l}
\hline \hline 
\multicolumn{2}{c}{{\bf Encoding strategy to create a codebook C}} \\
\hline \\
{\bf Public coin:} &
Choose $q^n \sim Q^n$. \\ \\
{\bf Sender 1:} &
For $m_{1c} = 1 \mbox{ to } 2^{nR_{1c}}$, choose 
$u^n(m_{1c}) \sim U^n|q^n$ independently.  \\ 
& 
For $m_{1p} = 1 \mbox{ to } 2^{nR_{1p}}$, choose 
$x^n(m_{1c},m_{1p}) \sim X^n|(q^n,u^n(m_{1c}))$ independently.  \\
& Feed $x^n(m_{1c},m_{1p})$ as input to the channel. 
\\ \\
{\bf Sender 2:} &
For $m_{2c} = 1 \mbox{ to } 2^{nR_{2c}}$, choose 
$v^n(m_{2c}) \sim V^n|q^n$ independently.  \\ 
&
For $m_{2p} = 1 \mbox{ to } 2^{nR_{2p}}$, choose 
$y^n(m_{2c},m_{1p}) \sim Y^n|(q^n,v^n(m_{2c}))$ independently. \\
& Feed $y^n(m_{2c},m_{2p})$ as input to the channel. 
\\ \\
\hline \hline
\end{tabular}
\end{center}
\caption{Encoding strategy for the ccqq-IC.}
\label{fig:ccqqencoding}
\end{figure}

Fix a typical $q^n \in T^{Q^n}_\delta$. Fix a symbol $q' \in \cQ$. 
Consider the copies $i$ of the channel $\chan$ with $q^n(i) = q'$.
From the typicality of $q^n$, the number of such copies satisfies
$N(q'|q^n) \in n p(q) (1 \pm \delta)$.
For those copies, from the point of view of each receiver, the channel 
behaves as a CMG-MAC with joint state 
\[
\rho^{UXV B_1}(q') := 
\sum_{(u,x,v) \in \cU \times \cX \times \cV} 
p_{UX|Q}(u,x|q') p_{V|Q}(v|q') \ketbra{u,x,v} \otimes 
\rho^{B_1}_{x,v},
\]
where
$
\rho^{B_1}_{x,v} :=
\sum_{y \in \cY} p_{Y|VQ}(y|v,q') \rho^{B_1}_{x,y}.
$
Using the decoding strategy of Section~\ref{subsec:cmgmac}, receiver~1
can achieve any rate triple in the following region:
\[
\begin{array}{c}
R_{2c}(q') < I(V : B_1 | X,q')_\rho, ~~~
R_{1p}(q') < I(X : B_1 | U V,q')_\rho, ~~~ 
R_{1c}(q') + R_{1p}(q') < I(X : B_1 | V,q')_\rho, \\
R_{1p}(q') + R_{2c}(q') < I(X V : B_1 | U,q')_\rho, ~~~~~
R_{1c}(q') + R_{1p}(q') + R_{2c}(q') < I(X V : B_1|q')_\rho, \\
\mbox{OR} \\
R_{2c}(q') \geq I(V : B_1 | X,q')_\rho, ~~~~~
R_{1p}(q') < I(V : B_1 | U,q')_\rho, ~~~~~ 
R_{1c}(q') + R_{1p}(q') < I(X : B_1 | q')_\rho. 
\end{array}
\]
Running over all possible symbols $q' \in q^n$ and appealing to the
typicality of $q^n$, we see that receiver~1 can decode with low
error probability for any point
in the following rate region:
\begin{equation}
\begin{array}{c}
R_{2c} < I(V : B_1 | X Q)_\rho, ~~~~~
R_{1p} < I(X : B_1 | U V Q)_\rho, ~~~~~ 
R_{1c} + R_{1p} < I(X : B_1 | V Q)_\rho, \\
R_{1p} + R_{2c} < I(X V : B_1 | U Q)_\rho, ~~~~~
R_{1c} + R_{1p} + R_{2c} < I(X V : B_1| Q)_\rho, \\
\mbox{OR} \\
R_{2c} \geq I(V : B_1 | X Q)_\rho, ~~~~~
R_{1p} < I(V : B_1 | U Q)_\rho, ~~~~~ 
R_{1c} + R_{1p} < I(X : B_1 |  Q)_\rho. 
\end{array}
\label{eq:cmg-ic1}
\end{equation}
Since the decoding error probability of receiver~1 is small, an
application of Fact~\ref{fact:gentle} allows Receiver~2 to pretend
as if he is the sole receiver, and use the same CMG-MAC decoding strategy.
Hence, he can decode with low error probability for any point in the
following rate region:
\begin{equation}
\begin{array}{c}
R_{1c} < I(U : B_2 | Y Q)_\rho, ~~~~~
R_{2p} < I(Y : B_2 | U V Q)_\rho, ~~~~~ 
R_{2c} + R_{2p} < I(Y : B_2 | U Q)_\rho, \\
R_{2p} + R_{1c} < I(Y U : B_2 | V Q)_\rho, ~~~~~
R_{2c} + R_{2p} + R_{1c} < I(Y U : B_2| Q)_\rho, \\
\mbox{OR} \\
R_{1c} \geq I(U : B_2 | Y Q)_\rho, ~~~~~
R_{2p} < I(U : B_2 | V Q)_\rho, ~~~~~ 
R_{2c} + R_{2p} < I(Y : B_2 |  Q)_\rho. 
\end{array}
\label{eq:cmg-ic2}
\end{equation}
We have thus shown that any rate quadruple 
$(R_{1c}, R_{1p}, R_{2c}, R_{2p})$ satisfying the inequalities
in inequalities~\ref{eq:cmg-ic1} and \ref{eq:cmg-ic2} above
gives us an achievable rate pair $(R_1, R_2)$ for the ccqq-IC, where
$R_1 = R_{1c} + R_{1p}$ and $R_2 = R_{2c} + R_{2p}$. 

Since our inner bound for the CMG-MAC is larger than the inner bound 
known previously in the classical setting, our inner bound for the 
interference channel is at least as large as the Chong-Motani-Garg
inner bound. One may wonder whether our inner bound is actually larger.
We now prove that in fact it is not so; for the interference channel
we get exactly the same inner bound as Chong, Motani and Garg. We are 
grateful to Omar Fawzi~\cite{fawzi:cmg} for this observation and include 
it here with his permission. 

Consider a rate quadruple $(R_{1c}, R_{1p}, R_{2c}, R_{2p})$ where
$(R_{1c}, R_{1p}, R_{2c})$ lies in the second part of the region in
inequalities~\ref{eq:cmg-ic1}. 
Recall that in this case Receiver~1 
ignores the common message of Sender~2 and decodes the 
complete message of Sender~1 by using projectors depending only on
the marginal distribution on $UX$. Consider now a different encoding 
strategy where the common message of Sender~2 is set to a fixed value,
that is, $V$ is now fixed,
$Y$ is distributed with the same marginal distribution as before and
$UX$ also have the same marginal distribution as before. Receiver~1
can decode with the same rate pair $(R_{1c}, R_{1p})$ as before using
the same strategy. Observe now that Receiver~2 can decode with the
new rate pair 
$(R'_{1c}, R'_{2c}, R'_{2p}) = (R_{1c}, 0, R_{2c} + R_{2p})$ since in
the inequalities~\ref{eq:cmg-ic2}, the constraints on $R_{1c}$,
$R'_{2c} + R'_{2p}$
and $R'_{2c} + R'_{2p} + R'_{1c}$ are unaffected by fixing $V$, and
the constraints on $R'_{2p}$ and $R'_{2p} + R_{1c}$ become equal to
those on $R'_{2c} + R'_{2p}$ and $R'_{2c} + R'_{2p} + R'_{1c}$
respectively after fixing $V$. Thus, the rates $R_1 = R_{1c} + R_{1p}$
and $R_2 = R_{2c} + R_{2p}$ are unaltered by the new encoding and
decoding strategy. In the new strategy the rate triple
$(R_{1c}, R_{1p}, 0)$ lies in the first part of the region
defined by the inequalities~\ref{eq:cmg-ic1}, whereas the rate triple
$(0, R_{2c} + R_{2p}, R_{1c})$ lies in the same part of the region
defined by the inequalities~\ref{eq:cmg-ic2} as before since
$R_{1c}$ is unaffected. We can repeat the argument with Receiver~2
allowing us to conclude that for any feasible rate quadruple
$(R_{1c}, R_{1p}, R_{2c}, R_{2p})$, there is a another feasible
rate quadruple $(R''_{1c}, R''_{1p}, R''_{2c}, R''_{2p})$ such that 
$R''_1 := R''_{1c} + R''_{1p} = R_{1c} + R_{1p} =: R_1$, 
$R''_2 := R''_{2c} + R''_{2p} = R_{2c} + R_{2p} =: R_2$, and 
$(R''_{1c}, R''_{1p}, R''_{2c})$  and $(R''_{1c},  R''_{2c}, R''_{2p})$
lie in 
the first parts of the regions defined by inequalities~\ref{eq:cmg-ic1}
and \ref{eq:cmg-ic2} respectively. Since the first parts of our regions
are contained in the Chong-Motani-Garg region, we conclude that we
get exactly the same inner bound for the interference channel as 
Chong, Motani and Garg.
In fact, what we have actually shown is
that the first parts of the regions defined 
by inequalities~\ref{eq:cmg-ic1} and \ref{eq:cmg-ic2} suffice to get the 
Chong-Motani-Garg inner bound for the interference channel.

\section{Discussion and open problems}
\label{sec:discussion}
In this paper, we have made several contributions to quantum Shannon
theory. We have constructed a general paradigm for sequentially decoding
channels with classical input and quantum output which directly mimics 
the classical intuition.
We have applied this paradigm successfully to several channels culminating
in a proof of the achievability of the Chong-Motani-Garg inner bound for
the two user interference channel. In order to do this, we discovered
a new kind of conditionally typical projector that mimics the operation
of intersection of two typical sets in classical information theory. 

Though we have proved the Chong-Motani-Garg inner bound for the ccqq-IC
by first proving an inner bound for the CMG-MAC as in the classical
approach of Chong, Motani and Garg, it is actually possible to
obtain the Chong-Motani-Garg inner bound using just our simultaneous
decoder for the ccq-MAC by combining it with geometric arguments from
\c{S}a\c{s}o\u{g}lu's paper~\cite{sasoglu}. Details will appear in the
full version of the paper. We followed the conventional route and
proved our inner bound for the CMG-MAC
in order to illustrate the reach and limitation of our current techniques.

The main problem left open in this work is the existence of a 
simultaneous decoder
for the multiple access channel with three or more senders. Such a 
decoder may pave the way for porting many results from classical network
information theory to the quantum setting.

\section*{Acknowledgements}
We are grateful to Omar Fawzi for allowing his observation
that the larger inner bound on the CMG-MAC does not lead to an
enlargement of the Chong-Motani-Garg inner bound for the interference
channel, to be included in the paper.
We thank Mark Wilde, Patrick Hayden and Ivan Savov
for going through preliminary drafts of this paper and giving 
valuable feedback. 

\bibliography{sequential}

\newcommand{\etalchar}[1]{$^{#1}$}
\begin{thebibliography}{CMGG08}

\bibitem[BSST02]{bennettentanglement}
C.~Bennett, P.~Shor, J.~Smolin, and A.~Thapliyal.
\newblock Entanglement-assisted capacity of a quantum channel and the reverse
  {Shannon} theorem.
\newblock {\em IEEE Transactions on Information Theory}, 48:2637--2655, 2002.

\bibitem[CG87]{costaelgamal}
M.~Costa and A.~El Gamal.
\newblock The capacity region of the discrete memoryless interference channel
  with strong interference.
\newblock {\em IEEE Transactions on Information Theory}, 33(5):710--711, 1987.

\bibitem[CMG06]{cmg}
H.~Chong, M.~Motani, and H.~Garg.
\newblock A comparison of two achievable rate regions for the interference
  channel.
\newblock In {\em Proceedings of the USCD-ITA Workshop}, San Diego, California,
  USA, February 2006.

\bibitem[CMGG08]{cmgelgamal}
H.~Chong, M.~Motani, H.~Garg, and H.~El Gamal.
\newblock On the {Han-Kobayashi} region for the interference channel.
\newblock {\em IEEE Transactions on Information Theory}, 54(7):3188--3195,
  2008.

\bibitem[\c{S}08]{sasoglu}
E.~\c{S}a\c{s}o\u{g}lu.
\newblock {Successive cancellation for cyclic interference channels}.
\newblock In {\em Proceedings of the IEEE Information Theory Workshop 2008},
  pages 36--40, Porto, Portugal, May 2008.

\bibitem[Faw11]{fawzi:cmg}
O.~Fawzi.
\newblock Private communication, July 2011.

\bibitem[GK10]{elgamalkim}
Abbas~El Gamal and Young-Han Kim.
\newblock Lecture notes on network information theory.
\newblock arXiv:/1001.3404, 2010.

\bibitem[GLM10]{lloyd:seq}
V.~Giovannetti, S.~Lloyd, and L.~Maccone.
\newblock Achieving the {Holevo} bound via sequential measurements.
\newblock Available at arXiv:quant-ph/1012.0386, December 2010.

\bibitem[HK81]{hk}
T.~Han and K.~Kobayashi.
\newblock A new achievable rate region for the interference channel.
\newblock {\em IEEE Transactions on Information Theory}, 27(1):49--60, January
  1981.

\bibitem[HN03]{hayashinagaoka}
M.~Hayashi and H.~Nagaoka.
\newblock General formulas for capacity of classical-quantum channels.
\newblock {\em IEEE Transactions on Information Theory}, 49(7):1753--1768,
  2003.

\bibitem[Hol98]{holevo}
A.~Holevo.
\newblock The capacity of the quantum channel with general signal states.
\newblock {\em IEEE Trans. Inf. Theory}, 44(1):269--273, 1998.
\newblock arXiv:quant-ph/9611023.

\bibitem[HW94]{pgm}
P.~Hausladen and W.~Wootters.
\newblock A pretty good measurement for distinguishing quantum states.
\newblock {\em Journal of Modern Optics}, 41(12):2385--2390, 1994.

\bibitem[MW05]{watrous:qma}
C.~Marriott and J.~Watrous.
\newblock Quantum {Arthur-Merlin} games.
\newblock {\em Computational Complexity}, 2005.

\bibitem[ON07]{ogawanagaoka:gentle}
T.~Ogawa and H.~Nagaoka.
\newblock Making good codes for classical-quantum channel coding via quantum
  hypothesis testing.
\newblock {\em IEEE Transactions on Information Theory}, 53(6):2261--2266, June
  2007.

\bibitem[Sch95]{schumacher}
B.~Schumacher.
\newblock Quantum coding.
\newblock {\em Physical Review A}, 51(4):2738--2747, April 1995.

\bibitem[SFW{\etalchar{+}}11]{mcgill:qic}
I.~Savov, O.~Fawzi, M.~Wilde, P.~Sen, and P.~Hayden.
\newblock Quantum interference channels.
\newblock In {\em Proceedings of the 49th Annual Allerton Conference on
  Communication, Control, and Computing}, 2011.
\newblock To appear. Also available at arXiv:quant-ph/1102.2955. Full version
  available at arXiv:quant-ph/1102.2624.

\bibitem[SW97]{schumacherwestmoreland}
B.~Schumacher and M.~Westmoreland.
\newblock Sending classical information via noisy quantum channels.
\newblock {\em Physical Review A}, 56(1):131--138, July 1997.

\bibitem[Sze04]{szegedy:quantumwalk}
M.~Szegedy.
\newblock Quantum speed-up of {Markov} chain based algorithms.
\newblock In {\em Proceedings of the 45th Annual IEEE Symposium on Foundations
  of Computer Science}, pages 32--41, 2004.

\bibitem[Wat09]{watrous:zk}
John Watrous.
\newblock Zero knowledge against quantum attacks.
\newblock {\em SIAM Journal on Computing}, 39(1):25--58, 2009.

\bibitem[Wil11]{wilde:book}
M.~Wilde.
\newblock From classical to quantum {Shannon} theory.
\newblock arXiv:1106.1445, June 2011.

\bibitem[Win99a]{winter:strongconverse}
A.~Winter.
\newblock Coding theorem and strong converse for quantum channels.
\newblock {\em IEEE Transactions on Information Theory}, 45(7):2481--2485,
  1999.

\bibitem[Win99b]{winterthesis}
A.~Winter.
\newblock {\em Coding Theorems of Quantum Information Theory}.
\newblock PhD thesis, Universit\"{a}t Bielefeld (arXiv:quant-ph/9907077), 1999.

\bibitem[Win01]{winter:mac}
A.~Winter.
\newblock The capacity of the quantum multiple-access channel.
\newblock {\em IEEE Transactions on Information Theory}, 47(7):3059--3065,
  2001.

\bibitem[XW11]{xuwilde:assisted}
S.~Xu and M.~Wilde.
\newblock Sequential, successive and simultaneous decoders for entanglement
  assisted classical communication.
\newblock arXiv:1107.1347, July 2011.

\end{thebibliography}

\end{document}